\documentclass[american,aps,pra,reprint,floatfix,nofootinbib,superscriptaddress]{revtex4-1}
\usepackage[unicode=true,pdfusetitle, bookmarks=true,bookmarksnumbered=false,bookmarksopen=false, breaklinks=false,pdfborder={0 0 0},backref=false,colorlinks=false]{hyperref}
\hypersetup{colorlinks,linkcolor=myurlcolor,citecolor=myurlcolor,urlcolor=myurlcolor}
\usepackage{graphics,epstopdf,graphicx,amsthm,amsmath,amssymb,mathptmx,braket,colortbl,color,bm,framed,mathrsfs}
\usepackage[T1]{fontenc}
\usepackage[up]{subfigure}
\usepackage{tikz}

\definecolor{myurlcolor}{rgb}{0,0,0.9}

\newcommand{\proj}[1]{| #1\rangle\!\langle #1 |}

\newcommand{\iinner}[2]{\langle #1 | #2\rangle}

\DeclareMathOperator{\trace}{Tr}
\newcommand{\Ptr}[2]{\trace_{#1}\Pa{#2}}
\newcommand{\Tr}[1]{\Ptr{}{#1}}
\newcommand{\Innerm}[3]{\left\langle #1 \left| #2 \right| #3 \right\rangle}
\newcommand{\Pa}[1]{\left[#1\right]}

\newcommand{\norm}[1]{\left\lVert #1 \right\rVert}

\theoremstyle{plain}
\newtheorem{thm}{Theorem}
\newtheorem{lem}[thm]{Lemma}
\newtheorem{prop}[thm]{Proposition}

\newcommand*{\myproofname}{Proof}
\newenvironment{mproof}[1][\myproofname]{\begin{proof}[#1]}{\end{proof}}
\def\ot{\otimes}

\def\real{\mathbb{R}}

\def\cI{\mathcal{I}}

\def\cM{\mathcal{M}}

\def\cD{\mathcal{D}}

\def\cH{\mathcal{H}}

\def\cE{\mathcal{E}}

\makeatother

\begin{document}

  \author{Kaifeng Bu}
 \email{bkf@zju.edn.cn}
 \affiliation{School of Mathematical Sciences, Zhejiang University, Hangzhou 310027, PR~China}
 \author{Uttam Singh}
  \email{uttamsingh@hri.res.in}
 \affiliation{Harish-Chandra Research Institute, Allahabad, 211019, India}
\affiliation{Homi Bhabha National Institute, Training School Complex, Anushakti Nagar,
     Mumbai 400085, India}
   \author{Shao-Ming Fei}
   \email{feishm@cnu.edn.cn}
 \affiliation{School of Mathematical Sciences, Capital Normal University, Beijing 100048, PR China}
  \affiliation{Max-Planck-Institute for Mathematics in the Sciences, 04103 Leipzig, Germany}
  \author{Arun Kumar Pati}
  \email{akpati@hri.res.in}
\affiliation{Harish-Chandra Research Institute, Allahabad, 211019, India}
  \author{Junde Wu}
   \email{wjd@zju.edn.cn}
 \affiliation{School of Mathematical Sciences, Zhejiang University, Hangzhou 310027, PR~China}

\title{Max- relative entropy of coherence : an operational coherence measure}

\begin{abstract}
The operational characterization of quantum coherence is the corner stone in the development of  resource theory of coherence.
We introduce a new coherence quantifier based on max-relative entropy. We prove that  max-relative entropy of coherence is directly related to the maximum overlap with maximally coherent states under a particular class of operations, which provides an operational interpretation of max-relative entropy of coherence. Moreover, we show that, for any coherent state, there are examples of subchannel discrimination problems such that this coherent state allows for a higher probability of successfully discriminating subchannels than that of all incoherent states.
This advantage of coherent states in subchannel discrimination can be exactly
characterized by the max-relative entropy of coherence.
By introducing suitable smooth max-relative entropy of coherence,
we prove that the smooth max-relative entropy of coherence
provides a lower bound of one-shot coherence cost, and the max-relative entropy of
coherence is equivalent to the relative entropy of coherence in asymptotic limit.
Similar to max-relative entropy of coherence, min-relative entropy of coherence has also been investigated. We show that the min-relative entropy of coherence provides an upper bound of
one-shot coherence distillation, and in asymptotic limit the min-relative
entropy of coherence is equivalent to the relative entropy of coherence.
\end{abstract}

\maketitle

\section{Introduction}
Quantumness in a single system is characterized by quantum coherence, namely, the superposition of
a state in a given reference basis. The coherence of a state may quantify the capacity of a system
in many quantum manipulations, ranging from  metrology \cite{Giovannetti2011} to thermodynamics \cite{Lostaglio2015,Lostaglio2015NC} .
Recently, various efforts have been made to develop a resource theory of coherence \cite{Baumgratz2014,Girolami2014,Streltsov2015,Winter2016,Killoran2016,Chitambar2016,Chitambar2016a}.
One of the earlier resource theories is that of quantum entanglement \cite{HorodeckiRMP09}, which is a basic resource for various quantum information processing protocols such as
superdense coding \cite{Bennett1992}, remote state preparation \cite{Pati2000,Bennett2001} and quantum teleportation \cite{Bennett1993}. Other notable examples include the resource theories of asymmetry \cite{Bartlett2007,Gour2008,Gour2009,Marvian2012,Marvian2013,Marvian14,Marvian2014}, thermodynamics \cite{Fernando2013}, and steering \cite{Rodrigo2015}. One of the main advantages that a resource theory offers is the lucid quantitative and operational description as well as the manipulation of the relevant resources at ones disposal, thus operational
characterization of quantum coherence is required in the resource theory of coherence.

A resource theory is usually composed of two basic elements: free states and free operations. The
set of allowed states (operations) under the given constraint is what we call the set of free states (operations).
Given a fixed basis, say $\set{\ket{i}}^{d-1}_{i=0}$ for a d-dimensional system, any quantum state which is diagonal
in the reference basis is called an incoherent state and is a free state in the resource theory of coherence. The set of
incoherent states is denoted by $\cI$. Any quantum state can be mapped into an incoherent state by a
full dephasing operation $\Delta$, where $\Delta(\rho):=\sum^{d-1}_{i=0}\Innerm{i}{\rho}{i}\proj{i}$.
However, there is no general consensus on the set of free operations in the resource theory of coherence.
We refer the following types of free operations in this work:
maximally incoherent operations (MIO) \cite{Chitambar2016b}, incoherent operations (IO) \cite{Baumgratz2014}, dephasing-covariant operations (DIO) \cite{Chitambar2016b} and strictly incoherent operations (SIO) \cite{Chitambar2016a,Chitambar2016b}.
By maximally incoherent operation (MIO), we refer to the maximal set of quantum operations
$\Phi$ which maps the incoherent states into incoherent states, i.e., $\Phi(\cI)\subset\cI$ \cite{Chitambar2016b}.
Incoherent operations (IO) is the set of all  quantum operations $\Phi$ that
admit a set of Kraus operators $\set{K_i}$ such that $\Phi(\cdot)=\sum_iK_i(\cdot)K^\dag_i$ and
$K_i\cI K^\dag_i\subset\cI$ for any $i$ \cite{Baumgratz2014}.
Dephasing-covariant operations (DIO) are the quantum operations $\Phi$ with
$[\Delta,\Phi]=0$ \cite{Chitambar2016b}. Strictly incoherent operations (SIO) is the set of all quantum operations
$\Phi$ admitting a set of Kraus operators $\set{K_i}$ such that $\Phi(\cdot)=\sum_iK_i(\cdot)K^\dag_i$ and
$\Delta (K_i\rho K^\dag_i)=K_i\Delta(\rho) K^\dag_i$ for any $i$ and any quantum state $\rho$.
Both IO and DIO are subsets of MIO , and SIO is a subset of both IO and DIO \cite{Chitambar2016b}.
However, IO and DIO are two different types of free operations and there is no inclusion relationship between
them (The operational gap between them can be seen in \cite{Bu2016}).

Several operational coherence quantifiers have been introduced as candidate coherence measures, subjecting to physical requirements such as monotonicity under certain type of free operations in the resource theory
of coherence. One canonical measure to quantify coherence is the relative entropy of coherence, which
is defined as $C_r(\rho)=S(\Delta(\rho))-S(\rho)$, where $S(\rho)=-\Tr{\rho\log\rho}$ is the von Neumann entropy \cite{Baumgratz2014}.
The relative entropy of coherence plays an important role in the process of coherence distillation, in which it can be interpreted as the optimal rate to distill
maximally coherent state from a given state $\rho$ by IO in the asymptotic limit \cite{Winter2016}.
Besides, the $l_1$ norm of coherence \cite{Baumgratz2014}, which is defined
as $C_{l_1}(\rho)=\sum_{i\neq j}|\rho_{ij}|$ with $\rho_{ij}=\Innerm{i}{\rho}{j}$, has also attracted
lots of discussions about its operational interpretation \cite{Rana2016}.
Recently, an operationally motivated coherence measure- robustness of coherence (RoC) -
has been introduced, which quantifies the minimal mixing required to erase the coherence in a given quantum state \cite{Napoli2016,Piani2016}.
There is growing concern about the operational characterization of quantum coherence
and further investigations are needed to provide an explicit and rigorous
operational interpretation of coherence.

In this letter, we introduce a new coherence measure based on max-relative entropy
and focus on its operational characterizations.
Max- and min- relative entropies have been introduced and investigated in \cite{Datta2009IEEE,Datta2009,Brandao2011,Buscemi2010}. The well-known (conditional and unconditional) max- and min- entropies  \cite{Renner2004IEEE,Renner2005phd}
can be obtained from these two quantities. It has been shown that max- and min-entropies are of operational
significance in the applications ranging from data compression \cite{Renes2012,Renner2004IEEE} to state merging \cite{Horodecki2007} and security of key \cite{Buhrman2006,Konig2009IEEE}. Besides, max- and min- relative entropies have been used to define entanglement monotone and their operational significance in the manipulation of entanglement has been provided in \cite{Datta2009IEEE,Datta2009,Brandao2011,Buscemi2010}.
Here, we define max-relative of coherence $C_{\max}$ based on max-relative entropy and investigate the properties
of $C_{\max}$. We prove that max-relative entropy of coherence for a given state $\rho$ is the maximum achievable
overlap with maximally coherent states under DIO, IO and SIO, which
gives rise to an operational interpretation of $C_{\max}$ and shows the equivalence among DIO, IO and SIO in an operational task.
Besides, we show that max-relative entropy of coherence characterizes the
role of  quantum states in an operational task: subchannel discrimination.
Subchannel discrimination is an important quantum information task which distinguishes the branches of a quantum evolution
for a quantum system to undergo \cite{Piani2015PRL}. It has been shown that every entangled or steerable
state is a resource in some instance of subchannel discrmination problems \cite{Piani2009,Piani2015PRL}. Here, we prove that
that every coherent state is useful in the subchannel discrimination of certain instruments, where the usefulness can be quantified by the
max-relative entropy of coherence of the given quantum state.
By smoothing the max-relative entropy of coherence, we introduce $\epsilon-$smoothed max-relative entropy of coherence
$C^{\epsilon}_{\max}$ for any fixed $\epsilon>0$ and show that the smooth max-relative entropy gives an lower bound of  coherence cost in  one-shot version. Moreover,
 we prove that for any quantum state,
max-relative entropy of coherence  is equivalent to
the relative entropy of coherence in asymptotic limit.

Corresponding to the max-relative entropy of coherence, we also introduce the min-relative entropy of coherence $C_{\min}$ by min-relative entropy, which is not a proper coherence
measure as it may increase on average under IO.
However, it gives an upper bound for the maximum overlap between the given states and the set of incoherent states. This implies that min-relative entropy of coherence
 also provides a lower bound of a well-known coherence measure, geometry of coherence \cite{Streltsov2015}.
By smoothing the min-relative entropy of coherence, we introduce $\epsilon-$smoothed min-relative entropy of coherence
$C^{\epsilon}_{\min}$ for any fixed $\epsilon>0$ and show that the smooth max-relative entropy gives an upper bound of  coherence distillation in  one-shot version.
Furthermore, we show that the min-relative of coherence is also
equivalent to distillation of coherence in asymptotic limit. The relationship among $C_{\min}$, $C_{\max}$ and other
coherence measures has also been investigated.

\section{Main results}
Let $\cH$ be a d-dimensional Hilbert space and $\cD(\cH)$ be the set
of density operators acting on $\cH$. Given two operators  $\rho$ and $\sigma$ with
$\rho\geq 0$, $\Tr{\rho}\leq 1$ and $\sigma\geq 0$,  the max-relative entropy of $\rho$ with respect to $\sigma$ is defined by \cite{Datta2009IEEE,Datta2009},
\begin{eqnarray}
D_{\max}(\rho||\sigma):&=&\min\set{\lambda:\rho\leq2^\lambda \sigma}.
\end{eqnarray}
We introduce a new coherence quantifier by max-relative  entropy: max-relative entropy of coherence $C_{\max}$,
\begin{eqnarray}\label{cmax}
C_{\max}(\rho):=\min_{\sigma\in \cI}D_{\max}(\rho||\sigma),
\end{eqnarray}
where $\cI$ is the set of incoherent states in $\cD(\cH)$.

We now show that $C_{\max}$ satisfies the conditions a coherence measure needs to fulfil. First, it is obvious that
$C_{\max}(\rho)\geq 0$. And since $D_{\max}(\rho||\sigma)=0$ iff $\rho=\sigma$ \cite{Datta2009IEEE}, we have $C_{\max}(\rho)=0$ if and only if $\rho\in\cI$.
Besides, as $D_{\max}$ is monotone under CPTP maps \cite{Datta2009IEEE},  we have
$C_{\max}(\Phi(\rho))\leq C_{\max}(\rho)$ for any incoherent operation $\Phi$.
Moreover, $C_{\max}$ is nonincreasing on average under incoherent operations, that is,
for any incoherent operation $\Phi(\cdot)=\sum_i K_i(\cdot)K^\dag_i$
with $K_i\cI K^\dag_i\subset\cI$,
$\sum_i p_iC_{\max}(\tilde{\rho}_i)\leq C_{\max}(\rho)$,
where $p_i=\Tr{K_i\rho K^\dag_i}$ and $\tilde{\rho}_i=K_i\rho K^\dag_i/p_i$, see proof in Supplemental Material \cite{supplement}.

{\it Remark} We have proven that the max-relative entropy of coherence $C_{\max}$ is a bona fide measure of coherence.
Since $D_{\max}$ is not jointly convex, we may not expect that $C_{\max}$ has the convexity, which is a desirable (although
not a fundamental) property for a coherence quantifier. However, we can prove that for $\rho=\sum^n_ip_i\rho_i$,
$C_{\max}(\rho)\leq\max_{i}C_{\max}(\rho_i)$.
Suppose that $C_{\max}(\rho_i)=D_{\max}(\rho_i||\sigma^*_i)$ for some $\sigma^*_i$, then
from the fact that  $D_{\max}(\sum_ip_i\rho_i||\sum_ip_i\sigma_i)\leq \max_iD_{\max}(\rho_i||\sigma_i)$ \cite{Datta2009IEEE}, we have
$C_{\max}(\rho)\leq D_{\max}(\sum_ip_i\rho_i||\sum_ip_i\sigma^*_i) \leq \max_iD_{\max}(\rho_i||\sigma^*_i)=\max_iC_{\max}(\rho_i)$. Besides, although
$C_{\max}$ is not convex, we can obtain a proper coherence measure with
convexity from $C_{\max}$ by the approach of convex roof extension, see Supplemental Material \cite{supplement}.

In the following, we concentrate on the operational characterization of the
max-relative entropy of coherence, and provide
operational interpretations of $C_{\max}$.

\textit{\textbf{Maximum overlap with maximally coherent states.}}---At first we show that
$2^{C_{\max}}$ is equal to the maximum overlap with the maximally coherent state that can be achieved by  DIO, IO and SIO.

\begin{thm}\label{thm:op_max}
 Given a quantum state $\rho\in\cD(\cH)$, we have
\begin{eqnarray}
2^{C_{\max}(\rho)}&=&d\max_{
\mathcal{E},\ket{\Psi}} F(\mathcal{E}(\rho),\proj{\Psi})^2,
\end{eqnarray}
where $F(\rho,\sigma)=\Tr{|\sqrt{\rho}\sqrt{\sigma}|}$ is the fidelity between states $\rho$ and $\sigma$ \cite{Nielsen10},
$\ket{\Psi}\in\cM$ and $\cM$ is the set of maximally coherent states in $\cD(\cH)$, $\mathcal{E}$ belongs to either
DIO or IO or SIO.
\end{thm}
(See proof in Supplemental Material \cite{supplement}.)

Here although IO, DIO and SIO are different types of free operations in resource theory of coherence \cite{Chitambar2016b,Bu2016},
they have the same behavior in the maximum overlap with the maximally coherent states.
From the view of coherence distillation \cite{Winter2016}, the maximum
overlap with maximally coherent states can be regarded as the distillation of coherence from
given states under IO, DIO and SIO. As fidelity can be used to define certain distance, thus
$C_{\max}(\rho)$  can also be viewed as the distance between the set of maximally coherent state and the
set of  $\set{\cE(\rho)}_{\cE\in\theta}$, where $\theta=  DIO, IO ~\text{or}~ SIO$.

Besides distillation of coherence, another kind of coherence manipulation is the coherence cost \cite{Winter2016}.
Now we study the one-shot version of coherence cost under MIO based on smooth max-relative entropy of coherence.
We define the one-shot coherence cost of a quantum state $\rho$ under MIO as
\begin{eqnarray*}
C^{(1),\epsilon}_{C,MIO}(\rho):=\min_{\substack{
\mathcal{E}\in MIO\\
M\in \mathbb{Z}
}}
\set{\log M: F(\rho,\mathcal{E}(\proj{\Psi^M_+}))^2\geq1-\epsilon
},
\end{eqnarray*}
where $\ket{\Psi^M_+}=\frac{1}{\sqrt{M}}\sum^M_{i=1}\ket{i}$, $\mathbb{Z}$ is the set of integer and
$\epsilon>0$. The $\epsilon$-smoothed max-relative entropy of coherence of a quantum state $\rho$ is defined by,
\begin{eqnarray}
C^{\epsilon}_{\max}(\rho):=\min_{\rho'\in B_{\epsilon}(\rho)}C_{\max}(\rho'),
\end{eqnarray}
where $B_\epsilon(\rho):=\set{\rho'\geq0:\norm{\rho'-\rho}_1\leq \epsilon, \Tr{\rho'}\leq\Tr{\rho}}$.
We find that the
smooth max-relative entropy of coherence gives a lower bound of one-shot coherence cost.
Given a quantum state $\rho\in\cD(\cH)$, for any $\epsilon>0$,
\begin{eqnarray}
C^{\epsilon'}_{\max}(\rho)\leq  C^{(1),\epsilon}_{C,MIO}(\rho),
\end{eqnarray}
where $\epsilon'=2\sqrt{\epsilon}$, see proof in Supplemental Material \cite{supplement}.

Besides, in view of smooth max-relative entropy of coherence, we can obtain the
equivalence between max-relative entropy of coherence and relative entropy of coherence in the asymptotic limit.
Since  relative entropy of coherence is the optimal rate to distill
maximally coherent state from a given state under certain free operations in the asymptotic limit \cite{Winter2016},
the smooth max-relative entropy of coherence in asymptotic limit is just the distillation of coherence. That is,
given a quantum state $\rho\in\cD(\cH)$, we have
\begin{eqnarray}\label{eq:max_vs_r}
\lim_{\epsilon\to 0}\lim_{n\to \infty}
\frac{1}{n}C^{\epsilon}_{\max}(\rho^{\ot n})
=C_r(\rho).
\end{eqnarray}
(The proof is presented in Supplemental Material \cite{supplement}.)

\textit{\textbf{Maximum advantage achievable in subchannel discrimination.}}--
Now, we investigate  another quantum information processing task: subchannel discrimination, which can also provide an operational
interpretation of $C_{\max}$. Subchannel discrimination is an important quantum information task which is used to identify the branch of a quantum evolution
to undergo. We consider some special instance of subchannel discrimination problem to show the advantage of coherent states.

A linear completely positive and trace non-increasing map $\cE$ is called a subchannel. If
a subchannel $\cE$ is trace preserving, then $\cE$ is called a channel.
An instrument $\mathfrak{I}=\set{\cE_a}_a$ for a channel $\cE$ is a collection of subchannels $\cE_a$ with
$\cE=\sum_a\cE_a$  and every instrument has its physical realization \cite{Piani2015PRL}. A dephasing covariant instrument $\mathfrak{I}^D$ for a DIO $\cE$
is a collection of subchannels $\set{\cE_a}_a$ such that $\cE=\sum_a\cE_a$. Similarly, we can define
incoherent instrument $\mathfrak{I}^I$ and strictly incoherent instrument $\mathfrak{I}^S$ for channel $\cE\in IO$ and $\cE\in SIO$ respectively.

Given an instrument $\mathfrak{I}=\set{\cE_a}_a$ for a quantum channel $\cE$, let us consider a Positive Operator Valued Measurement (POVM)
$\set{M_b}_b$ with $\sum_bM_b=\mathbb{I}$. The probability of successfully discriminating the subchannels in the instrument
$\mathfrak{I}$ by POVM $\set{M_b}_b$ for input state $\rho$ is given by
\begin{eqnarray}
p_{\text{succ}}(\mathfrak{I},\set{M_b}_b, \rho)
=\sum_a\Tr{\cE_a(\rho)M_a}.
\end{eqnarray}
The optimal probability of success in subchannel discrimination of $\mathfrak{I}$ over
all POVMs is given by
\begin{eqnarray}
p_{\text{succ}}(\mathfrak{I},\rho)=\max_{\set{M_b}_b}
p_{\text{succ}}(\mathfrak{I},\set{M_b}_b, \rho).
\end{eqnarray}
If we restrict the input states to be incoherent ones, then the optimal probability of success among all incoherent states is given by
\begin{eqnarray}
p^{ICO}_{\text{succ}}(\mathfrak{I})=\max_{\sigma\in\cI}p_{\text{succ}}(\mathfrak{I},\sigma).
\end{eqnarray}
We have the following theorem.

\begin{thm}\label{thm:op_sub}
Given a quantum state $\rho$, $2^{C_{\max}(\rho)}$
is the maximal advantage achievable by $\rho$
compared with incoherent states in all subschannel discrimination problems
of dephasing-covariant, incoherent and strictly incoherent instruments,
\begin{eqnarray}
2^{C_{\max}(\rho)}&=&\max_{\mathfrak{I}}
\frac{p_{\text{succ}}(\mathfrak{I},\rho)}{p^{ICO}_{\text{succ}}(\mathfrak{I})},
\end{eqnarray}
where $\mathfrak{I}$ is either $\mathfrak{I}^D$ or $\mathfrak{I}^I$ or $\mathfrak{I}^S$, denoting the dephasing-covariant, incoherent and strictly incoherent instrument, respectively.
\end{thm}
The proof of Theorem \ref{thm:op_sub} is presented in Supplemental Material \cite{supplement}. This result shows that the advantage of coherent states in
certain instances of subchannel discrimination problems can be exactly captured by $C_{\max}$, which provides another operational interpretation
of $C_{\max}$ and also shows the equivalence among DIO, IO and SIO in the information processing task of subchannel discrimination.

\textit{\textbf{Min-relative entropy of coherence $C_{\min}(\rho)$.}}--Given  two operators  $\rho$ and $\sigma$ with
$\rho\geq 0, \Tr{\rho}\leq 1$ and $\sigma\geq 0$,  max- and min- relative entropy of $\rho$ relative to $\sigma$ are defined as
\begin{eqnarray}
D_{\min}(\rho||\sigma):=-\log\Tr{\Pi_{\rho}\sigma}
\end{eqnarray}
where $\Pi_{\rho}$ denotes the projector onto $\text{supp} \rho$, the  support of $\rho$.
Corresponding to $C_{\max}(\rho)$ defined in (\ref{cmax}), we can similarly introduce a quantity defined by min-relative entropy,
\begin{eqnarray}
C_{\min}(\rho):=\min_{\sigma\in\cI}D_{\min}(\rho||\sigma).
\end{eqnarray}
Since $D_{\min}(\rho||\sigma)=0$ if $\text{supp}\,\rho=\text{supp}\,\sigma$ \cite{Datta2009IEEE},
we have
$\rho\in\cI~\Rightarrow ~C_{\min}(\rho)=0$. However, converse direction may not be true, for example,
let $\rho=\frac{1}{2}\proj{0}+\frac{1}{2}\proj{+}$ with $\ket{+}=
\frac{1}{\sqrt{2}}(\ket{1}+\ket{2})$, then $\rho$ is coherent but $C_{\min}(\rho)=0$.
Besides, as $D_{\min}$ is monotone under CPTP maps \cite{Datta2009IEEE},
we have  $C_{\min}(\Phi(\rho))\leq C_{\min}(\rho)$ for any
$\Phi \in IO$.
However, $C_{\min}$ may increase on average under IO (see Supplemental Material \cite{supplement}). Thus, $C_{\min}$ is not be a proper coherence measure as $C_{\max}$.

Although $C_{\min}$ is not a good coherence quantifier, it still has some interesting properties in the manipulation of coherence. First, $C_{\min}$ gives upper bound of the maximum overlap with the set of incoherent states for any given quantum state $\rho\in\cD(\cH)$,
\begin{eqnarray}\label{ineq:min_G}
2^{-C_{\min}(\rho)}\geq\max_{\sigma\in\cI}F(\rho,\sigma)^2.
\end{eqnarray}
Moreover, if $\rho$ is pure state $\ket{\psi}$, then above equality holds, that is,
\begin{eqnarray}
2^{-C_{\min}(\psi)}=\max_{\sigma\in\cI}F(\psi,\sigma)^2,
\end{eqnarray}
see proof in  Supplemental Material \cite{supplement}.

Moreover, for geometry of coherence defined by $C_g(\rho)=1- \max_{\sigma\in\cI}F(\rho,\sigma)^2$ \cite{Streltsov2015}, $C_{\min}$ also provides a lower bound for $C_g$ as follows
 \begin{eqnarray}
 C_g(\rho)\geq 1-2^{-C_{\min}(\rho)}.
 \end{eqnarray}

Now let us consider again the one-shot version of distillable coherence under MIO by
modifying and smoothing the min-relative entropy of coherence $C_{\min}$. We define the one-shot distillable coherence of a quantum state $\rho$ under MIO as
\begin{eqnarray*}
 C^{(1),\epsilon}_{D,MIO}(\rho):=\max_{\substack{
\mathcal{E}\in MIO\\
M\in \mathbb{Z}
}}\set{\log M: F(\mathcal{E}(\rho),\proj{\Psi^M_+})^2\geq1-\epsilon},
\end{eqnarray*}
where $\ket{\Psi^M_+}=\frac{1}{\sqrt{M}}\sum^M_{i=1}\ket{i}$ and $\epsilon>0$.

For any $\epsilon>0$, we define the smooth min-relative entropy of coherence of a quantum state $\rho$ as follows
\begin{eqnarray}
C^{\epsilon}_{\min}(\rho):=\max_{\substack{
0\leq A\leq \mathbb{I}\\
\Tr{A\rho}\geq 1-\epsilon}}
\min_{\sigma\in\cI}-\log\Tr{A\sigma},
\end{eqnarray}
where $\mathbb{I}$ denotes the
identity.
It can be shown that $C^{\epsilon}_{\min}$ is a upper bound of one-shot distillable coherence,
\begin{eqnarray}\label{ineq:min_dis}
  C^{(1),\epsilon}_{D,MIO}(\rho)\leq C^{\epsilon}_{\min}(\rho)
\end{eqnarray}
for any $\epsilon>0$, see proof in Supplemental Material \cite{supplement}.

The distillation of coherence in asymptotic limit can be expressed as
\begin{eqnarray*}
C_{D,MIO}=\lim_{\epsilon\to 0}\lim_{n\to\infty}\frac{1}{n}C^{(1),\epsilon}_{D,MIO}(\rho).
\end{eqnarray*}
It has been proven that $C_{D,MIO}(\rho)=C_r(\rho)$ \cite{Winter2016}. Here we show that the equality in inequality \eqref{ineq:min_dis} holds in the asymptotic limit as the $C_{\min}$ is equivalent to $C_r$ in the asymptotic limit.
Given a quantum state $\rho\in\cD(\cH)$, then
\begin{eqnarray}\label{eq:min_vs_r}
\lim_{\epsilon\to 0}\lim_{n\to \infty}
\frac{1}{n}C^{\epsilon}_{\min}(\rho^{\ot n})
=C_r(\rho).
\end{eqnarray}
(The proof is presented in Supplemental Material \cite{supplement}.)

We have shown that $C_{\min}$ gives rise to the bounds for maximum overlap with the incoherent states
and for one-shot distillable coherence. Indeed the exact expression of $C_{\min}$ for some special class of quantum states can be calculated.
For pure state $\ket{\psi}=\sum^d_{i=1}\psi_i\ket{i}$ with $\sum^d_{i=1}|\psi_i|^2=1$, we have
$C_{\min}(\psi)=-\log\max_{i}|\psi_i|^2$. For maximally coherent state $\ket{\Psi}=\frac{1}{\sqrt{d}}\sum^d_{j=1}e^{i\theta_j}\ket{j}$, we have
$C_{\min}(\Psi)=\log d$, which is the maximum value for $C_{\min}$ in d-dimensional space.

\textit{\textbf{Relationship between $C_{\max}$ and other coherence measures.}}--
First, we investigate the relationship among $C_{\max}$, $C_{\min}$ and $C_r$.
Since $D_{\min}(\rho||\sigma)\leq S(\rho||\sigma)\leq D_{\max}(\rho||\sigma)$  for any quantum states $\rho$ and $\sigma$ \cite{Datta2009IEEE},
one has
\begin{eqnarray}
C_{\min}(\rho)\leq C_r(\rho)\leq C_{\max}(\rho).
\end{eqnarray}
Moreover, as mentioned before,
these quantities are all equal in the asymptotic limit.

Above all, $C_{\max}$ is equal to the logarithm of robustness of coherence, as
$RoC(\rho)=\min_{\sigma\in\cI}\set{s\geq0|\rho\leq(1+s)\sigma}$ and $C_{\max}(\rho)=\min_{\sigma\in \cI}\min\set{\lambda:\rho\leq2^\lambda \sigma}$ \cite{Chitambar2016b}, that is,  $2^{C_{\max}(\rho)}=1+RoC(\rho)$.
Thus, the operational interpretations of $C_{\max}$ in terms of maximum overlap with maximally coherent states and
subchannel discrimination, can also be viewed as the operational interpretations of robustness of coherence $RoC$.
It is known that robustness of coherence plays an important role in a phase discrimination task, which provides an
operational interpretation for robustness of coherence \cite{Napoli2016}. This phase discrimination task
 investigated in \cite{Napoli2016} is just
a special case of the subchannel discrimination in depasing-covariant  instruments. Due to the relationship between $C_{\max}$ and $RoC$, we can obtain the closed form of
$C_{\max}$ for some special class of quantum states.
As an example, let us consider a pure state $\ket{\psi}=\sum^d_{i=1}\psi_i\ket{i}$. Then $C_{\max}(\psi)=\log((\sum^d_{i=1}|\psi_i|)^2))=2\log(\sum^d_{i=1}|\psi_i|)$.
Thus, for maximally coherent state $\ket{\Psi}=\frac{1}{\sqrt{d}}\sum^d_{j=1}e^{i\theta_j}\ket{j}$, we have
$C_{\max}(\Psi)=\log d$, which is the maximum value for $C_{\max}$ in d-dimensional space.

Since $RoC(\rho)\leq C_{l_1}(\rho)$ \cite{Napoli2016} and $1+RoC(\rho)=2^{C_{\max}(\rho)}$, then
$C_{\max}(\rho)\leq \log(1+C_{l_1}(\rho))$. We have the relationship among these coherence measures,
\begin{eqnarray*}
C_{\min}(\rho)\leq C_r(\rho)\leq C_{\max}(\rho)
 &=&\log(1+RoC(\rho))\\
&\leq& \log(1+C_{l_1}(\rho)),~~~~
\end{eqnarray*}
which implies that $2^{C_r(\rho)}\leq 1+C_{l_1}(\rho)$ (See also \cite{Rana2016}).

\section{Conclusion}\label{sec:con}
We have investigated the properties of max- and min-relative entropy of coherence, especially the operational interpretation of the max-relative entropy of coherence.
It has been found that the max-relative entropy of coherence characterizes the
maximum overlap with the maximally coherent states under DIO, IO and SIO, as
well as the maximum advantage achievable by coherent states compared with all incoherent states in
subchannel discrimination problems of all dephasing-covariant, incoherent and strictly incoherent instruments, which
also provides new operational interpretations of robustness of coherence and illustrates the equivalence of DIO, IO and SIO in these two
operational taks. The study of $C_{\max}$ and $C_{\min}$ also makes the relationship between the operational
coherence measures (e.g. $C_r$ and $C_{l_1}$ ) more clear.
These results may highlight the understanding to the operational resource theory of coherence.

Besides, the relationships among smooth max- and min- relative relative entropy of coherence
and one-shot coherence cost and distillation have been investigated explicitly. As
both smooth max- and min- relative entropy of coherence
are equal to relative entropy of coherence in the asymptotic limit and the significance of relative entropy of coherence in the distillation of coherence, further studies are desired on the one-shot coherence cost
and distillation.

\begin{acknowledgments}
This work is supported by the Natural Science Foundation of China (Grants No. 11171301, No. 10771191, No. 11571307 and No. 11675113) and the Doctoral Programs Foundation of the Ministry of Education of China (Grant No. J20130061).
\end{acknowledgments}

\bibliographystyle{apsrev4-1}
 \bibliography{Maxcoh-lit}

\begin{thebibliography}{54}%
\makeatletter
\providecommand \@ifxundefined [1]{%
 \@ifx{#1\undefined}
}%
\providecommand \@ifnum [1]{%
 \ifnum #1\expandafter \@firstoftwo
 \else \expandafter \@secondoftwo
 \fi
}%
\providecommand \@ifx [1]{%
 \ifx #1\expandafter \@firstoftwo
 \else \expandafter \@secondoftwo
 \fi
}%
\providecommand \natexlab [1]{#1}%
\providecommand \enquote  [1]{``#1''}%
\providecommand \bibnamefont  [1]{#1}%
\providecommand \bibfnamefont [1]{#1}%
\providecommand \citenamefont [1]{#1}%
\providecommand \href@noop [0]{\@secondoftwo}%
\providecommand \href [0]{\begingroup \@sanitize@url \@href}%
\providecommand \@href[1]{\@@startlink{#1}\@@href}%
\providecommand \@@href[1]{\endgroup#1\@@endlink}%
\providecommand \@sanitize@url [0]{\catcode `\\12\catcode `\$12\catcode
  `\&12\catcode `\#12\catcode `\^12\catcode `\_12\catcode `\%12\relax}%
\providecommand \@@startlink[1]{}%
\providecommand \@@endlink[0]{}%
\providecommand \url  [0]{\begingroup\@sanitize@url \@url }%
\providecommand \@url [1]{\endgroup\@href {#1}{\urlprefix }}%
\providecommand \urlprefix  [0]{URL }%
\providecommand \Eprint [0]{\href }%
\providecommand \doibase [0]{http://dx.doi.org/}%
\providecommand \selectlanguage [0]{\@gobble}%
\providecommand \bibinfo  [0]{\@secondoftwo}%
\providecommand \bibfield  [0]{\@secondoftwo}%
\providecommand \translation [1]{[#1]}%
\providecommand \BibitemOpen [0]{}%
\providecommand \bibitemStop [0]{}%
\providecommand \bibitemNoStop [0]{.\EOS\space}%
\providecommand \EOS [0]{\spacefactor3000\relax}%
\providecommand \BibitemShut  [1]{\csname bibitem#1\endcsname}%
\let\auto@bib@innerbib\@empty
\bibitem [{\citenamefont {Giovannetti}\ \emph {et~al.}(2011)\citenamefont
  {Giovannetti}, \citenamefont {Lloyd},\ and\ \citenamefont
  {Maccone}}]{Giovannetti2011}%
  \BibitemOpen
  \bibfield  {author} {\bibinfo {author} {\bibfnamefont {V.}~\bibnamefont
  {Giovannetti}}, \bibinfo {author} {\bibfnamefont {S.}~\bibnamefont {Lloyd}},
  \ and\ \bibinfo {author} {\bibfnamefont {L.}~\bibnamefont {Maccone}},\ }\href
  {http://dx.doi.org/10.1038/nphoton.2011.35} {\bibfield  {journal} {\bibinfo
  {journal} {Nat. Photon.}\ }\textbf {\bibinfo {volume} {5}},\ \bibinfo {pages}
  {222} (\bibinfo {year} {2011})}\BibitemShut {NoStop}%
\bibitem [{\citenamefont {Lostaglio}\ \emph
  {et~al.}(2015{\natexlab{a}})\citenamefont {Lostaglio}, \citenamefont
  {Korzekwa}, \citenamefont {Jennings},\ and\ \citenamefont
  {Rudolph}}]{Lostaglio2015}%
  \BibitemOpen
  \bibfield  {author} {\bibinfo {author} {\bibfnamefont {M.}~\bibnamefont
  {Lostaglio}}, \bibinfo {author} {\bibfnamefont {K.}~\bibnamefont {Korzekwa}},
  \bibinfo {author} {\bibfnamefont {D.}~\bibnamefont {Jennings}}, \ and\
  \bibinfo {author} {\bibfnamefont {T.}~\bibnamefont {Rudolph}},\ }\href
  {\doibase 10.1103/PhysRevX.5.021001} {\bibfield  {journal} {\bibinfo
  {journal} {Phys. Rev. X}\ }\textbf {\bibinfo {volume} {5}},\ \bibinfo {pages}
  {021001} (\bibinfo {year} {2015}{\natexlab{a}})}\BibitemShut {NoStop}%
\bibitem [{\citenamefont {Lostaglio}\ \emph
  {et~al.}(2015{\natexlab{b}})\citenamefont {Lostaglio}, \citenamefont
  {Jennings},\ and\ \citenamefont {Rudolph}}]{Lostaglio2015NC}%
  \BibitemOpen
  \bibfield  {author} {\bibinfo {author} {\bibfnamefont {M.}~\bibnamefont
  {Lostaglio}}, \bibinfo {author} {\bibfnamefont {D.}~\bibnamefont {Jennings}},
  \ and\ \bibinfo {author} {\bibfnamefont {T.}~\bibnamefont {Rudolph}},\ }\href
  {http://dx.doi.org/10.1038/ncomms7383} {\bibfield  {journal} {\bibinfo
  {journal} {Nat. Commun.}\ }\textbf {\bibinfo {volume} {6}} (\bibinfo {year}
  {2015}{\natexlab{b}})}\BibitemShut {NoStop}%
\bibitem [{\citenamefont {Baumgratz}\ \emph {et~al.}(2014)\citenamefont
  {Baumgratz}, \citenamefont {Cramer},\ and\ \citenamefont
  {Plenio}}]{Baumgratz2014}%
  \BibitemOpen
  \bibfield  {author} {\bibinfo {author} {\bibfnamefont {T.}~\bibnamefont
  {Baumgratz}}, \bibinfo {author} {\bibfnamefont {M.}~\bibnamefont {Cramer}}, \
  and\ \bibinfo {author} {\bibfnamefont {M.~B.}\ \bibnamefont {Plenio}},\
  }\href {\doibase 10.1103/PhysRevLett.113.140401} {\bibfield  {journal}
  {\bibinfo  {journal} {Phys. Rev. Lett.}\ }\textbf {\bibinfo {volume} {113}},\
  \bibinfo {pages} {140401} (\bibinfo {year} {2014})}\BibitemShut {NoStop}%
\bibitem [{\citenamefont {Girolami}(2014)}]{Girolami2014}%
  \BibitemOpen
  \bibfield  {author} {\bibinfo {author} {\bibfnamefont {D.}~\bibnamefont
  {Girolami}},\ }\href {\doibase 10.1103/PhysRevLett.113.170401} {\bibfield
  {journal} {\bibinfo  {journal} {Phys. Rev. Lett.}\ }\textbf {\bibinfo
  {volume} {113}},\ \bibinfo {pages} {170401} (\bibinfo {year}
  {2014})}\BibitemShut {NoStop}%
\bibitem [{\citenamefont {Streltsov}\ \emph {et~al.}(2015)\citenamefont
  {Streltsov}, \citenamefont {Singh}, \citenamefont {Dhar}, \citenamefont
  {Bera},\ and\ \citenamefont {Adesso}}]{Streltsov2015}%
  \BibitemOpen
  \bibfield  {author} {\bibinfo {author} {\bibfnamefont {A.}~\bibnamefont
  {Streltsov}}, \bibinfo {author} {\bibfnamefont {U.}~\bibnamefont {Singh}},
  \bibinfo {author} {\bibfnamefont {H.~S.}\ \bibnamefont {Dhar}}, \bibinfo
  {author} {\bibfnamefont {M.~N.}\ \bibnamefont {Bera}}, \ and\ \bibinfo
  {author} {\bibfnamefont {G.}~\bibnamefont {Adesso}},\ }\href {\doibase
  10.1103/PhysRevLett.115.020403} {\bibfield  {journal} {\bibinfo  {journal}
  {Phys. Rev. Lett.}\ }\textbf {\bibinfo {volume} {115}},\ \bibinfo {pages}
  {020403} (\bibinfo {year} {2015})}\BibitemShut {NoStop}%
\bibitem [{\citenamefont {Winter}\ and\ \citenamefont
  {Yang}(2016)}]{Winter2016}%
  \BibitemOpen
  \bibfield  {author} {\bibinfo {author} {\bibfnamefont {A.}~\bibnamefont
  {Winter}}\ and\ \bibinfo {author} {\bibfnamefont {D.}~\bibnamefont {Yang}},\
  }\href {\doibase 10.1103/PhysRevLett.116.120404} {\bibfield  {journal}
  {\bibinfo  {journal} {Phys. Rev. Lett.}\ }\textbf {\bibinfo {volume} {116}},\
  \bibinfo {pages} {120404} (\bibinfo {year} {2016})}\BibitemShut {NoStop}%
\bibitem [{\citenamefont {Killoran}\ \emph {et~al.}(2016)\citenamefont
  {Killoran}, \citenamefont {Steinhoff},\ and\ \citenamefont
  {Plenio}}]{Killoran2016}%
  \BibitemOpen
  \bibfield  {author} {\bibinfo {author} {\bibfnamefont {N.}~\bibnamefont
  {Killoran}}, \bibinfo {author} {\bibfnamefont {F.~E.~S.}\ \bibnamefont
  {Steinhoff}}, \ and\ \bibinfo {author} {\bibfnamefont {M.~B.}\ \bibnamefont
  {Plenio}},\ }\href {\doibase 10.1103/PhysRevLett.116.080402} {\bibfield
  {journal} {\bibinfo  {journal} {Phys. Rev. Lett.}\ }\textbf {\bibinfo
  {volume} {116}},\ \bibinfo {pages} {080402} (\bibinfo {year}
  {2016})}\BibitemShut {NoStop}%
\bibitem [{\citenamefont {Chitambar}\ \emph {et~al.}(2016)\citenamefont
  {Chitambar}, \citenamefont {Streltsov}, \citenamefont {Rana}, \citenamefont
  {Bera}, \citenamefont {Adesso},\ and\ \citenamefont
  {Lewenstein}}]{Chitambar2016}%
  \BibitemOpen
  \bibfield  {author} {\bibinfo {author} {\bibfnamefont {E.}~\bibnamefont
  {Chitambar}}, \bibinfo {author} {\bibfnamefont {A.}~\bibnamefont
  {Streltsov}}, \bibinfo {author} {\bibfnamefont {S.}~\bibnamefont {Rana}},
  \bibinfo {author} {\bibfnamefont {M.~N.}\ \bibnamefont {Bera}}, \bibinfo
  {author} {\bibfnamefont {G.}~\bibnamefont {Adesso}}, \ and\ \bibinfo {author}
  {\bibfnamefont {M.}~\bibnamefont {Lewenstein}},\ }\href {\doibase
  10.1103/PhysRevLett.116.070402} {\bibfield  {journal} {\bibinfo  {journal}
  {Phys. Rev. Lett.}\ }\textbf {\bibinfo {volume} {116}},\ \bibinfo {pages}
  {070402} (\bibinfo {year} {2016})}\BibitemShut {NoStop}%
\bibitem [{\citenamefont {Chitambar}\ and\ \citenamefont
  {Gour}(2016{\natexlab{a}})}]{Chitambar2016a}%
  \BibitemOpen
  \bibfield  {author} {\bibinfo {author} {\bibfnamefont {E.}~\bibnamefont
  {Chitambar}}\ and\ \bibinfo {author} {\bibfnamefont {G.}~\bibnamefont
  {Gour}},\ }\href {\doibase 10.1103/PhysRevLett.117.030401} {\bibfield
  {journal} {\bibinfo  {journal} {Phys. Rev. Lett.}\ }\textbf {\bibinfo
  {volume} {117}},\ \bibinfo {pages} {030401} (\bibinfo {year}
  {2016}{\natexlab{a}})}\BibitemShut {NoStop}%
\bibitem [{\citenamefont {Horodecki}\ \emph {et~al.}(2009)\citenamefont
  {Horodecki}, \citenamefont {Horodecki}, \citenamefont {Horodecki},\ and\
  \citenamefont {Horodecki}}]{HorodeckiRMP09}%
  \BibitemOpen
  \bibfield  {author} {\bibinfo {author} {\bibfnamefont {R.}~\bibnamefont
  {Horodecki}}, \bibinfo {author} {\bibfnamefont {P.}~\bibnamefont
  {Horodecki}}, \bibinfo {author} {\bibfnamefont {M.}~\bibnamefont
  {Horodecki}}, \ and\ \bibinfo {author} {\bibfnamefont {K.}~\bibnamefont
  {Horodecki}},\ }\href {\doibase 10.1103/RevModPhys.81.865} {\bibfield
  {journal} {\bibinfo  {journal} {Rev. Mod. Phys.}\ }\textbf {\bibinfo {volume}
  {81}},\ \bibinfo {pages} {865} (\bibinfo {year} {2009})}\BibitemShut
  {NoStop}%
\bibitem [{\citenamefont {Bennett}\ and\ \citenamefont
  {Wiesner}(1992)}]{Bennett1992}%
  \BibitemOpen
  \bibfield  {author} {\bibinfo {author} {\bibfnamefont {C.~H.}\ \bibnamefont
  {Bennett}}\ and\ \bibinfo {author} {\bibfnamefont {S.~J.}\ \bibnamefont
  {Wiesner}},\ }\href {\doibase 10.1103/PhysRevLett.69.2881} {\bibfield
  {journal} {\bibinfo  {journal} {Phys. Rev. Lett.}\ }\textbf {\bibinfo
  {volume} {69}},\ \bibinfo {pages} {2881} (\bibinfo {year}
  {1992})}\BibitemShut {NoStop}%
\bibitem [{\citenamefont {Pati}(2000)}]{Pati2000}%
  \BibitemOpen
  \bibfield  {author} {\bibinfo {author} {\bibfnamefont {A.~K.}\ \bibnamefont
  {Pati}},\ }\href {\doibase 10.1103/PhysRevA.63.014302} {\bibfield  {journal}
  {\bibinfo  {journal} {Phys. Rev. A}\ }\textbf {\bibinfo {volume} {63}},\
  \bibinfo {pages} {014302} (\bibinfo {year} {2000})}\BibitemShut {NoStop}%
\bibitem [{\citenamefont {Bennett}\ \emph {et~al.}(2001)\citenamefont
  {Bennett}, \citenamefont {DiVincenzo}, \citenamefont {Shor}, \citenamefont
  {Smolin}, \citenamefont {Terhal},\ and\ \citenamefont
  {Wootters}}]{Bennett2001}%
  \BibitemOpen
  \bibfield  {author} {\bibinfo {author} {\bibfnamefont {C.~H.}\ \bibnamefont
  {Bennett}}, \bibinfo {author} {\bibfnamefont {D.~P.}\ \bibnamefont
  {DiVincenzo}}, \bibinfo {author} {\bibfnamefont {P.~W.}\ \bibnamefont
  {Shor}}, \bibinfo {author} {\bibfnamefont {J.~A.}\ \bibnamefont {Smolin}},
  \bibinfo {author} {\bibfnamefont {B.~M.}\ \bibnamefont {Terhal}}, \ and\
  \bibinfo {author} {\bibfnamefont {W.~K.}\ \bibnamefont {Wootters}},\ }\href
  {\doibase 10.1103/PhysRevLett.87.077902} {\bibfield  {journal} {\bibinfo
  {journal} {Phys. Rev. Lett.}\ }\textbf {\bibinfo {volume} {87}},\ \bibinfo
  {pages} {077902} (\bibinfo {year} {2001})}\BibitemShut {NoStop}%
\bibitem [{\citenamefont {Bennett}\ \emph {et~al.}(1993)\citenamefont
  {Bennett}, \citenamefont {Brassard}, \citenamefont {Cr\'epeau}, \citenamefont
  {Jozsa}, \citenamefont {Peres},\ and\ \citenamefont
  {Wootters}}]{Bennett1993}%
  \BibitemOpen
  \bibfield  {author} {\bibinfo {author} {\bibfnamefont {C.~H.}\ \bibnamefont
  {Bennett}}, \bibinfo {author} {\bibfnamefont {G.}~\bibnamefont {Brassard}},
  \bibinfo {author} {\bibfnamefont {C.}~\bibnamefont {Cr\'epeau}}, \bibinfo
  {author} {\bibfnamefont {R.}~\bibnamefont {Jozsa}}, \bibinfo {author}
  {\bibfnamefont {A.}~\bibnamefont {Peres}}, \ and\ \bibinfo {author}
  {\bibfnamefont {W.~K.}\ \bibnamefont {Wootters}},\ }\href {\doibase
  10.1103/PhysRevLett.70.1895} {\bibfield  {journal} {\bibinfo  {journal}
  {Phys. Rev. Lett.}\ }\textbf {\bibinfo {volume} {70}},\ \bibinfo {pages}
  {1895} (\bibinfo {year} {1993})}\BibitemShut {NoStop}%
\bibitem [{\citenamefont {Bartlett}\ \emph {et~al.}(2007)\citenamefont
  {Bartlett}, \citenamefont {Rudolph},\ and\ \citenamefont
  {Spekkens}}]{Bartlett2007}%
  \BibitemOpen
  \bibfield  {author} {\bibinfo {author} {\bibfnamefont {S.~D.}\ \bibnamefont
  {Bartlett}}, \bibinfo {author} {\bibfnamefont {T.}~\bibnamefont {Rudolph}}, \
  and\ \bibinfo {author} {\bibfnamefont {R.~W.}\ \bibnamefont {Spekkens}},\
  }\href {\doibase 10.1103/RevModPhys.79.555} {\bibfield  {journal} {\bibinfo
  {journal} {Rev. Mod. Phys.}\ }\textbf {\bibinfo {volume} {79}},\ \bibinfo
  {pages} {555} (\bibinfo {year} {2007})}\BibitemShut {NoStop}%
\bibitem [{\citenamefont {Gour}\ and\ \citenamefont
  {Spekkens}(2008)}]{Gour2008}%
  \BibitemOpen
  \bibfield  {author} {\bibinfo {author} {\bibfnamefont {G.}~\bibnamefont
  {Gour}}\ and\ \bibinfo {author} {\bibfnamefont {R.~W.}\ \bibnamefont
  {Spekkens}},\ }\href {http://stacks.iop.org/1367-2630/10/i=3/a=033023}
  {\bibfield  {journal} {\bibinfo  {journal} {New J. Phys.}\ }\textbf {\bibinfo
  {volume} {10}},\ \bibinfo {pages} {033023} (\bibinfo {year}
  {2008})}\BibitemShut {NoStop}%
\bibitem [{\citenamefont {Gour}\ \emph {et~al.}(2009)\citenamefont {Gour},
  \citenamefont {Marvian},\ and\ \citenamefont {Spekkens}}]{Gour2009}%
  \BibitemOpen
  \bibfield  {author} {\bibinfo {author} {\bibfnamefont {G.}~\bibnamefont
  {Gour}}, \bibinfo {author} {\bibfnamefont {I.}~\bibnamefont {Marvian}}, \
  and\ \bibinfo {author} {\bibfnamefont {R.~W.}\ \bibnamefont {Spekkens}},\
  }\href {\doibase 10.1103/PhysRevA.80.012307} {\bibfield  {journal} {\bibinfo
  {journal} {Phys. Rev. A}\ }\textbf {\bibinfo {volume} {80}},\ \bibinfo
  {pages} {012307} (\bibinfo {year} {2009})}\BibitemShut {NoStop}%
\bibitem [{\citenamefont {Marvian}(2012)}]{Marvian2012}%
  \BibitemOpen
  \bibfield  {author} {\bibinfo {author} {\bibfnamefont {I.}~\bibnamefont
  {Marvian}},\ }\href@noop {} {\emph {\bibinfo {title} {Symmetry, Asymmetry and
  Quantum Information}}}\ (\bibinfo  {publisher} {PhD thesis, University of
  Waterloo},\ \bibinfo {year} {2012})\BibitemShut {NoStop}%
\bibitem [{\citenamefont {Marvian}\ and\ \citenamefont
  {Spekkens}(2013)}]{Marvian2013}%
  \BibitemOpen
  \bibfield  {author} {\bibinfo {author} {\bibfnamefont {I.}~\bibnamefont
  {Marvian}}\ and\ \bibinfo {author} {\bibfnamefont {R.~W.}\ \bibnamefont
  {Spekkens}},\ }\href {http://stacks.iop.org/1367-2630/15/i=3/a=033001}
  {\bibfield  {journal} {\bibinfo  {journal} {New J. Phys.}\ }\textbf {\bibinfo
  {volume} {15}},\ \bibinfo {pages} {033001} (\bibinfo {year}
  {2013})}\BibitemShut {NoStop}%
\bibitem [{\citenamefont {Marvian}\ and\ \citenamefont
  {Spekkens}(2014{\natexlab{a}})}]{Marvian14}%
  \BibitemOpen
  \bibfield  {author} {\bibinfo {author} {\bibfnamefont {I.}~\bibnamefont
  {Marvian}}\ and\ \bibinfo {author} {\bibfnamefont {R.~W.}\ \bibnamefont
  {Spekkens}},\ }\href {\doibase 10.1103/PhysRevA.90.062110} {\bibfield
  {journal} {\bibinfo  {journal} {Phys. Rev. A}\ }\textbf {\bibinfo {volume}
  {90}},\ \bibinfo {pages} {062110} (\bibinfo {year}
  {2014}{\natexlab{a}})}\BibitemShut {NoStop}%
\bibitem [{\citenamefont {Marvian}\ and\ \citenamefont
  {Spekkens}(2014{\natexlab{b}})}]{Marvian2014}%
  \BibitemOpen
  \bibfield  {author} {\bibinfo {author} {\bibfnamefont {I.}~\bibnamefont
  {Marvian}}\ and\ \bibinfo {author} {\bibfnamefont {R.~W.}\ \bibnamefont
  {Spekkens}},\ }\href {http://dx.doi.org/10.1038/ncomms4821} {\bibfield
  {journal} {\bibinfo  {journal} {Nat. Commun.}\ }\textbf {\bibinfo {volume}
  {5}},\ \bibinfo {pages} {3821} (\bibinfo {year}
  {2014}{\natexlab{b}})}\BibitemShut {NoStop}%
\bibitem [{\citenamefont {Brand\~ao}\ \emph {et~al.}(2013)\citenamefont
  {Brand\~ao}, \citenamefont {Horodecki}, \citenamefont {Oppenheim},
  \citenamefont {Renes},\ and\ \citenamefont {Spekkens}}]{Fernando2013}%
  \BibitemOpen
  \bibfield  {author} {\bibinfo {author} {\bibfnamefont {F.~G. S.~L.}\
  \bibnamefont {Brand\~ao}}, \bibinfo {author} {\bibfnamefont {M.}~\bibnamefont
  {Horodecki}}, \bibinfo {author} {\bibfnamefont {J.}~\bibnamefont
  {Oppenheim}}, \bibinfo {author} {\bibfnamefont {J.~M.}\ \bibnamefont
  {Renes}}, \ and\ \bibinfo {author} {\bibfnamefont {R.~W.}\ \bibnamefont
  {Spekkens}},\ }\href {\doibase 10.1103/PhysRevLett.111.250404} {\bibfield
  {journal} {\bibinfo  {journal} {Phys. Rev. Lett.}\ }\textbf {\bibinfo
  {volume} {111}},\ \bibinfo {pages} {250404} (\bibinfo {year}
  {2013})}\BibitemShut {NoStop}%
\bibitem [{\citenamefont {Gallego}\ and\ \citenamefont
  {Aolita}(2015)}]{Rodrigo2015}%
  \BibitemOpen
  \bibfield  {author} {\bibinfo {author} {\bibfnamefont {R.}~\bibnamefont
  {Gallego}}\ and\ \bibinfo {author} {\bibfnamefont {L.}~\bibnamefont
  {Aolita}},\ }\href {\doibase 10.1103/PhysRevX.5.041008} {\bibfield  {journal}
  {\bibinfo  {journal} {Phys. Rev. X}\ }\textbf {\bibinfo {volume} {5}},\
  \bibinfo {pages} {041008} (\bibinfo {year} {2015})}\BibitemShut {NoStop}%
\bibitem [{\citenamefont {Chitambar}\ and\ \citenamefont
  {Gour}(2016{\natexlab{b}})}]{Chitambar2016b}%
  \BibitemOpen
  \bibfield  {author} {\bibinfo {author} {\bibfnamefont {E.}~\bibnamefont
  {Chitambar}}\ and\ \bibinfo {author} {\bibfnamefont {G.}~\bibnamefont
  {Gour}},\ }\href {\doibase 10.1103/PhysRevA.94.052336} {\bibfield  {journal}
  {\bibinfo  {journal} {Phys. Rev. A}\ }\textbf {\bibinfo {volume} {94}},\
  \bibinfo {pages} {052336} (\bibinfo {year} {2016}{\natexlab{b}})}\BibitemShut
  {NoStop}%
\bibitem [{\citenamefont {Bu}\ and\ \citenamefont {Xiong}()}]{Bu2016}%
  \BibitemOpen
  \bibfield  {author} {\bibinfo {author} {\bibfnamefont {K.}~\bibnamefont
  {Bu}}\ and\ \bibinfo {author} {\bibfnamefont {C.}~\bibnamefont {Xiong}},\
  }\href@noop {} {}\Eprint {http://arxiv.org/abs/1604.06524} {arXiv:1604.06524}
  \BibitemShut {NoStop}%
\bibitem [{\citenamefont {Rana}\ \emph {et~al.}()\citenamefont {Rana},
  \citenamefont {Parashar}, \citenamefont {Winter},\ and\ \citenamefont
  {Lewenstein}}]{Rana2016}%
  \BibitemOpen
  \bibfield  {author} {\bibinfo {author} {\bibfnamefont {S.}~\bibnamefont
  {Rana}}, \bibinfo {author} {\bibfnamefont {P.}~\bibnamefont {Parashar}},
  \bibinfo {author} {\bibfnamefont {A.}~\bibnamefont {Winter}}, \ and\ \bibinfo
  {author} {\bibfnamefont {M.}~\bibnamefont {Lewenstein}},\ }\href@noop {}
  {}\Eprint {http://arxiv.org/abs/1612.09234} {arXiv:1612.09234} \BibitemShut
  {NoStop}%
\bibitem [{\citenamefont {Napoli}\ \emph {et~al.}(2016)\citenamefont {Napoli},
  \citenamefont {Bromley}, \citenamefont {Cianciaruso}, \citenamefont {Piani},
  \citenamefont {Johnston},\ and\ \citenamefont {Adesso}}]{Napoli2016}%
  \BibitemOpen
  \bibfield  {author} {\bibinfo {author} {\bibfnamefont {C.}~\bibnamefont
  {Napoli}}, \bibinfo {author} {\bibfnamefont {T.~R.}\ \bibnamefont {Bromley}},
  \bibinfo {author} {\bibfnamefont {M.}~\bibnamefont {Cianciaruso}}, \bibinfo
  {author} {\bibfnamefont {M.}~\bibnamefont {Piani}}, \bibinfo {author}
  {\bibfnamefont {N.}~\bibnamefont {Johnston}}, \ and\ \bibinfo {author}
  {\bibfnamefont {G.}~\bibnamefont {Adesso}},\ }\href {\doibase
  10.1103/PhysRevLett.116.150502} {\bibfield  {journal} {\bibinfo  {journal}
  {Phys. Rev. Lett.}\ }\textbf {\bibinfo {volume} {116}},\ \bibinfo {pages}
  {150502} (\bibinfo {year} {2016})}\BibitemShut {NoStop}%
\bibitem [{\citenamefont {Piani}\ \emph {et~al.}(2016)\citenamefont {Piani},
  \citenamefont {Cianciaruso}, \citenamefont {Bromley}, \citenamefont {Napoli},
  \citenamefont {Johnston},\ and\ \citenamefont {Adesso}}]{Piani2016}%
  \BibitemOpen
  \bibfield  {author} {\bibinfo {author} {\bibfnamefont {M.}~\bibnamefont
  {Piani}}, \bibinfo {author} {\bibfnamefont {M.}~\bibnamefont {Cianciaruso}},
  \bibinfo {author} {\bibfnamefont {T.~R.}\ \bibnamefont {Bromley}}, \bibinfo
  {author} {\bibfnamefont {C.}~\bibnamefont {Napoli}}, \bibinfo {author}
  {\bibfnamefont {N.}~\bibnamefont {Johnston}}, \ and\ \bibinfo {author}
  {\bibfnamefont {G.}~\bibnamefont {Adesso}},\ }\href {\doibase
  10.1103/PhysRevA.93.042107} {\bibfield  {journal} {\bibinfo  {journal} {Phys.
  Rev. A}\ }\textbf {\bibinfo {volume} {93}},\ \bibinfo {pages} {042107}
  (\bibinfo {year} {2016})}\BibitemShut {NoStop}%
\bibitem [{\citenamefont {Datta}(2009{\natexlab{a}})}]{Datta2009IEEE}%
  \BibitemOpen
  \bibfield  {author} {\bibinfo {author} {\bibfnamefont {N.}~\bibnamefont
  {Datta}},\ }\href {\doibase 10.1109/TIT.2009.2018325} {\bibfield  {journal}
  {\bibinfo  {journal} {IEEE Trans. Inf. Theory}\ }\textbf {\bibinfo {volume}
  {55}},\ \bibinfo {pages} {2816} (\bibinfo {year}
  {2009}{\natexlab{a}})}\BibitemShut {NoStop}%
\bibitem [{\citenamefont {Datta}(2009{\natexlab{b}})}]{Datta2009}%
  \BibitemOpen
  \bibfield  {author} {\bibinfo {author} {\bibfnamefont {N.}~\bibnamefont
  {Datta}},\ }\href {\doibase 10.1142/S0219749909005298} {\bibfield  {journal}
  {\bibinfo  {journal} {Int. J. Quant. Inf.}\ }\textbf {\bibinfo {volume}
  {07}},\ \bibinfo {pages} {475} (\bibinfo {year}
  {2009}{\natexlab{b}})}\BibitemShut {NoStop}%
\bibitem [{\citenamefont {Brandao}\ and\ \citenamefont
  {Datta}(2011)}]{Brandao2011}%
  \BibitemOpen
  \bibfield  {author} {\bibinfo {author} {\bibfnamefont {F.~G. S.~L.}\
  \bibnamefont {Brandao}}\ and\ \bibinfo {author} {\bibfnamefont
  {N.}~\bibnamefont {Datta}},\ }\href {\doibase 10.1109/TIT.2011.2104531}
  {\bibfield  {journal} {\bibinfo  {journal} {IEEE Trans. Inf. Theory}\
  }\textbf {\bibinfo {volume} {57}},\ \bibinfo {pages} {1754} (\bibinfo {year}
  {2011})}\BibitemShut {NoStop}%
\bibitem [{\citenamefont {Buscemi}\ and\ \citenamefont
  {Datta}(2010)}]{Buscemi2010}%
  \BibitemOpen
  \bibfield  {author} {\bibinfo {author} {\bibfnamefont {F.}~\bibnamefont
  {Buscemi}}\ and\ \bibinfo {author} {\bibfnamefont {N.}~\bibnamefont
  {Datta}},\ }\href {\doibase 10.1109/TIT.2009.2039166} {\bibfield  {journal}
  {\bibinfo  {journal} {IEEE Trans. Inf. Theory}\ }\textbf {\bibinfo {volume}
  {56}},\ \bibinfo {pages} {1447} (\bibinfo {year} {2010})}\BibitemShut
  {NoStop}%
\bibitem [{\citenamefont {Renner}\ and\ \citenamefont
  {Wolf}(2004)}]{Renner2004IEEE}%
  \BibitemOpen
  \bibfield  {author} {\bibinfo {author} {\bibfnamefont {R.}~\bibnamefont
  {Renner}}\ and\ \bibinfo {author} {\bibfnamefont {S.}~\bibnamefont {Wolf}},\
  }in\ \href {\doibase 10.1109/ISIT.2004.1365269} {\emph {\bibinfo {booktitle}
  {Proc. Int. Symp. Inf. Theory}}}\ (\bibinfo {year} {2004})\ p.\ \bibinfo
  {pages} {233}\BibitemShut {NoStop}%
\bibitem [{\citenamefont {Renner}(2005)}]{Renner2005phd}%
  \BibitemOpen
  \bibfield  {author} {\bibinfo {author} {\bibfnamefont {R.}~\bibnamefont
  {Renner}},\ }\href@noop {} {\emph {\bibinfo {title} {Security of quantum key
  distribution}}}\ (\bibinfo  {publisher} {ETH Zurich},\ \bibinfo {year}
  {2005})\BibitemShut {NoStop}%
\bibitem [{\citenamefont {Renes}\ and\ \citenamefont
  {Renner}(2012)}]{Renes2012}%
  \BibitemOpen
  \bibfield  {author} {\bibinfo {author} {\bibfnamefont {J.~M.}\ \bibnamefont
  {Renes}}\ and\ \bibinfo {author} {\bibfnamefont {R.}~\bibnamefont {Renner}},\
  }\href {\doibase 10.1109/TIT.2011.2177589} {\bibfield  {journal} {\bibinfo
  {journal} {IEEE Trans. Inf. Theory}\ }\textbf {\bibinfo {volume} {58}},\
  \bibinfo {pages} {1985} (\bibinfo {year} {2012})}\BibitemShut {NoStop}%
\bibitem [{\citenamefont {Horodecki}\ \emph {et~al.}(2007)\citenamefont
  {Horodecki}, \citenamefont {Oppenheim},\ and\ \citenamefont
  {Winter}}]{Horodecki2007}%
  \BibitemOpen
  \bibfield  {author} {\bibinfo {author} {\bibfnamefont {M.}~\bibnamefont
  {Horodecki}}, \bibinfo {author} {\bibfnamefont {J.}~\bibnamefont
  {Oppenheim}}, \ and\ \bibinfo {author} {\bibfnamefont {A.}~\bibnamefont
  {Winter}},\ }\href {\doibase 10.1007/s00220-006-0118-x} {\bibfield  {journal}
  {\bibinfo  {journal} {Commun. Math. Phys.}\ }\textbf {\bibinfo {volume}
  {269}},\ \bibinfo {pages} {107} (\bibinfo {year} {2007})}\BibitemShut
  {NoStop}%
\bibitem [{\citenamefont {Buhrman}\ \emph {et~al.}(2006)\citenamefont
  {Buhrman}, \citenamefont {Christandl}, \citenamefont {Hayden}, \citenamefont
  {Lo},\ and\ \citenamefont {Wehner}}]{Buhrman2006}%
  \BibitemOpen
  \bibfield  {author} {\bibinfo {author} {\bibfnamefont {H.}~\bibnamefont
  {Buhrman}}, \bibinfo {author} {\bibfnamefont {M.}~\bibnamefont {Christandl}},
  \bibinfo {author} {\bibfnamefont {P.}~\bibnamefont {Hayden}}, \bibinfo
  {author} {\bibfnamefont {H.-K.}\ \bibnamefont {Lo}}, \ and\ \bibinfo {author}
  {\bibfnamefont {S.}~\bibnamefont {Wehner}},\ }\href {\doibase
  10.1103/PhysRevLett.97.250501} {\bibfield  {journal} {\bibinfo  {journal}
  {Phys. Rev. Lett.}\ }\textbf {\bibinfo {volume} {97}},\ \bibinfo {pages}
  {250501} (\bibinfo {year} {2006})}\BibitemShut {NoStop}%
\bibitem [{\citenamefont {Konig}\ \emph {et~al.}(2009)\citenamefont {Konig},
  \citenamefont {Renner},\ and\ \citenamefont {Schaffner}}]{Konig2009IEEE}%
  \BibitemOpen
  \bibfield  {author} {\bibinfo {author} {\bibfnamefont {R.}~\bibnamefont
  {Konig}}, \bibinfo {author} {\bibfnamefont {R.}~\bibnamefont {Renner}}, \
  and\ \bibinfo {author} {\bibfnamefont {C.}~\bibnamefont {Schaffner}},\ }\href
  {\doibase 10.1109/TIT.2009.2025545} {\bibfield  {journal} {\bibinfo
  {journal} {IEEE Trans. Inf. Theory}\ }\textbf {\bibinfo {volume} {55}},\
  \bibinfo {pages} {4337} (\bibinfo {year} {2009})}\BibitemShut {NoStop}%
\bibitem [{\citenamefont {Piani}\ and\ \citenamefont
  {Watrous}(2015)}]{Piani2015PRL}%
  \BibitemOpen
  \bibfield  {author} {\bibinfo {author} {\bibfnamefont {M.}~\bibnamefont
  {Piani}}\ and\ \bibinfo {author} {\bibfnamefont {J.}~\bibnamefont
  {Watrous}},\ }\href {\doibase 10.1103/PhysRevLett.114.060404} {\bibfield
  {journal} {\bibinfo  {journal} {Phys. Rev. Lett.}\ }\textbf {\bibinfo
  {volume} {114}},\ \bibinfo {pages} {060404} (\bibinfo {year}
  {2015})}\BibitemShut {NoStop}%
\bibitem [{\citenamefont {Piani}\ and\ \citenamefont
  {Watrous}(2009)}]{Piani2009}%
  \BibitemOpen
  \bibfield  {author} {\bibinfo {author} {\bibfnamefont {M.}~\bibnamefont
  {Piani}}\ and\ \bibinfo {author} {\bibfnamefont {J.}~\bibnamefont
  {Watrous}},\ }\href {\doibase 10.1103/PhysRevLett.102.250501} {\bibfield
  {journal} {\bibinfo  {journal} {Phys. Rev. Lett.}\ }\textbf {\bibinfo
  {volume} {102}},\ \bibinfo {pages} {250501} (\bibinfo {year}
  {2009})}\BibitemShut {NoStop}%
\bibitem [{sup()}]{supplement}%
  \BibitemOpen
  \href@noop {} {}\bibinfo {note} {See Supplemental Material [url] for the
  details of the proof, which includes Refs.
  \cite{Vedral1998,semidefinite1996,watrous2009,brandao2016quantum,Peng2015,Audenaert2007,Ogawa2000,Du2015,Hayashi2006,Brandao2010CMP,FBrandao15}.}\BibitemShut
  {Stop}%
\bibitem [{\citenamefont {Nielsen}\ and\ \citenamefont
  {Chuang}(2010)}]{Nielsen10}%
  \BibitemOpen
  \bibfield  {author} {\bibinfo {author} {\bibfnamefont {M.~A.}\ \bibnamefont
  {Nielsen}}\ and\ \bibinfo {author} {\bibfnamefont {I.~L.}\ \bibnamefont
  {Chuang}},\ }\href {http://dx.doi.org/10.1017/CBO9780511976667} {\emph
  {\bibinfo {title} {Quantum Computation and Quantum Information}}}\ (\bibinfo
  {publisher} {Cambridge University Press},\ \bibinfo {year}
  {2010})\BibitemShut {NoStop}%
\bibitem [{\citenamefont {Vedral}\ and\ \citenamefont
  {Plenio}(1998)}]{Vedral1998}%
  \BibitemOpen
  \bibfield  {author} {\bibinfo {author} {\bibfnamefont {V.}~\bibnamefont
  {Vedral}}\ and\ \bibinfo {author} {\bibfnamefont {M.~B.}\ \bibnamefont
  {Plenio}},\ }\href {\doibase 10.1103/PhysRevA.57.1619} {\bibfield  {journal}
  {\bibinfo  {journal} {Phys. Rev. A}\ }\textbf {\bibinfo {volume} {57}},\
  \bibinfo {pages} {1619} (\bibinfo {year} {1998})}\BibitemShut {NoStop}%
\bibitem [{\citenamefont {Vandenberghe}\ and\ \citenamefont
  {Boyd}(1996)}]{semidefinite1996}%
  \BibitemOpen
  \bibfield  {author} {\bibinfo {author} {\bibfnamefont {L.}~\bibnamefont
  {Vandenberghe}}\ and\ \bibinfo {author} {\bibfnamefont {S.}~\bibnamefont
  {Boyd}},\ }\href {http://dx.doi.org/10.1137/1038003} {\bibfield  {journal}
  {\bibinfo  {journal} {SIAM review}\ }\textbf {\bibinfo {volume} {38}},\
  \bibinfo {pages} {49} (\bibinfo {year} {1996})}\BibitemShut {NoStop}%
\bibitem [{\citenamefont {Watrous}(2009)}]{watrous2009}%
  \BibitemOpen
  \bibfield  {author} {\bibinfo {author} {\bibfnamefont {J.}~\bibnamefont
  {Watrous}},\ }\href {\doibase 10.4086/toc.2009.v005a011} {\bibfield
  {journal} {\bibinfo  {journal} {Theory of Computing}\ }\textbf {\bibinfo
  {volume} {5}},\ \bibinfo {pages} {217} (\bibinfo {year} {2009})}\BibitemShut
  {NoStop}%
\bibitem [{\citenamefont {Brandao}\ and\ \citenamefont
  {Svore}()}]{brandao2016quantum}%
  \BibitemOpen
  \bibfield  {author} {\bibinfo {author} {\bibfnamefont {F.~G.}\ \bibnamefont
  {Brandao}}\ and\ \bibinfo {author} {\bibfnamefont {K.}~\bibnamefont
  {Svore}},\ }\href {https://arxiv.org/abs/1609.05537} {}\Eprint
  {http://arxiv.org/abs/1609.05537} {arXiv:1609.05537} \BibitemShut {NoStop}%
\bibitem [{\citenamefont {Peng}\ \emph {et~al.}(2016)\citenamefont {Peng},
  \citenamefont {Jiang},\ and\ \citenamefont {Fan}}]{Peng2015}%
  \BibitemOpen
  \bibfield  {author} {\bibinfo {author} {\bibfnamefont {Y.}~\bibnamefont
  {Peng}}, \bibinfo {author} {\bibfnamefont {Y.}~\bibnamefont {Jiang}}, \ and\
  \bibinfo {author} {\bibfnamefont {H.}~\bibnamefont {Fan}},\ }\href {\doibase
  10.1103/PhysRevA.93.032326} {\bibfield  {journal} {\bibinfo  {journal} {Phys.
  Rev. A}\ }\textbf {\bibinfo {volume} {93}},\ \bibinfo {pages} {032326}
  (\bibinfo {year} {2016})}\BibitemShut {NoStop}%
\bibitem [{\citenamefont {Audenaert}(2007)}]{Audenaert2007}%
  \BibitemOpen
  \bibfield  {author} {\bibinfo {author} {\bibfnamefont {K.~M.~R.}\
  \bibnamefont {Audenaert}},\ }\href
  {http://stacks.iop.org/1751-8121/40/i=28/a=S18} {\bibfield  {journal}
  {\bibinfo  {journal} {J. Phys. A Math. Theor.}\ }\textbf {\bibinfo {volume}
  {40}},\ \bibinfo {pages} {8127} (\bibinfo {year} {2007})}\BibitemShut
  {NoStop}%
\bibitem [{\citenamefont {Ogawa}\ and\ \citenamefont
  {Nagaoka}(2000)}]{Ogawa2000}%
  \BibitemOpen
  \bibfield  {author} {\bibinfo {author} {\bibfnamefont {T.}~\bibnamefont
  {Ogawa}}\ and\ \bibinfo {author} {\bibfnamefont {H.}~\bibnamefont
  {Nagaoka}},\ }\href {\doibase 10.1109/18.887855} {\bibfield  {journal}
  {\bibinfo  {journal} {IEEE Trans. Inf. Theory}\ }\textbf {\bibinfo {volume}
  {46}},\ \bibinfo {pages} {2428} (\bibinfo {year} {2000})}\BibitemShut
  {NoStop}%
\bibitem [{\citenamefont {Du}\ \emph {et~al.}()\citenamefont {Du},
  \citenamefont {Bai},\ and\ \citenamefont {Qi}}]{Du2015}%
  \BibitemOpen
  \bibfield  {author} {\bibinfo {author} {\bibfnamefont {S.}~\bibnamefont
  {Du}}, \bibinfo {author} {\bibfnamefont {Z.}~\bibnamefont {Bai}}, \ and\
  \bibinfo {author} {\bibfnamefont {X.}~\bibnamefont {Qi}},\ }\href@noop {}
  {}\Eprint {http://arxiv.org/abs/1504.02862} {arXiv:1504.02862} \BibitemShut
  {NoStop}%
\bibitem [{\citenamefont {Hayashi}(2006)}]{Hayashi2006}%
  \BibitemOpen
  \bibfield  {author} {\bibinfo {author} {\bibfnamefont {M.}~\bibnamefont
  {Hayashi}},\ }\href@noop {} {\emph {\bibinfo {title} {Quantum Information:An
  introduction}}}\ (\bibinfo  {publisher} {Springer},\ \bibinfo {year}
  {2006})\BibitemShut {NoStop}%
\bibitem [{\citenamefont {Brand{\~a}o}\ and\ \citenamefont
  {Plenio}(2010)}]{Brandao2010CMP}%
  \BibitemOpen
  \bibfield  {author} {\bibinfo {author} {\bibfnamefont {F.~G. S.~L.}\
  \bibnamefont {Brand{\~a}o}}\ and\ \bibinfo {author} {\bibfnamefont {M.~B.}\
  \bibnamefont {Plenio}},\ }\href {\doibase 10.1007/s00220-010-1005-z}
  {\bibfield  {journal} {\bibinfo  {journal} {Commun. Math. Phys.}\ }\textbf
  {\bibinfo {volume} {295}},\ \bibinfo {pages} {791} (\bibinfo {year}
  {2010})}\BibitemShut {NoStop}%
\bibitem [{\citenamefont {Brand\~ao}\ and\ \citenamefont
  {Gour}(2015)}]{FBrandao15}%
  \BibitemOpen
  \bibfield  {author} {\bibinfo {author} {\bibfnamefont {F.~G. S.~L.}\
  \bibnamefont {Brand\~ao}}\ and\ \bibinfo {author} {\bibfnamefont
  {G.}~\bibnamefont {Gour}},\ }\href {\doibase 10.1103/PhysRevLett.115.070503}
  {\bibfield  {journal} {\bibinfo  {journal} {Phys. Rev. Lett.}\ }\textbf
  {\bibinfo {volume} {115}},\ \bibinfo {pages} {070503} (\bibinfo {year}
  {2015})}\BibitemShut {NoStop}%
\end{thebibliography}%

\appendix
\section{strong monotonicity under IO for $C_{\max}$ }\label{apen:max_mon}

We prove this property based on the method in  \cite{Vedral1998} and the basic facts of $D_{\max}$ \cite{Datta2009IEEE}. Due to the definition of $C_{\max}$, there exists an optimal $\sigma_*\in\cI$ such that $C_{\max}(\rho)=D_{\max}(\rho||\sigma_*)$. Let $\tilde{\sigma}_i=K_i\sigma_* K^\dag_i/\Tr{K_i\sigma_* K^\dag_i}$, then we have
\begin{eqnarray*}
&&\sum_ip_i D_{\max}(\tilde{\rho}_i||\tilde{\sigma}_i)\\
&\leq& \sum_i D_{\max}(K_i\rho K^\dag_i||K_i\sigma_* K^\dag_i)\\
&\leq& \sum_i D_{\max}(\Ptr{E}{\mathbb{I}\ot \proj{i}U\rho\ot\proj{\alpha}U^\dag \mathbb{I}\ot \proj{i}}||\\
&&\times\Ptr{E}{\mathbb{I}\ot \proj{i}U\sigma_*\ot\proj{\alpha}U^\dag \mathbb{I}\ot \proj{i}})\\
&\leq&\sum_i D_{\max}(\mathbb{I}\ot \proj{i}U\rho\ot\proj{\alpha}U^\dag \mathbb{I}\ot \proj{i}||\\
&&\times \mathbb{I}\ot \proj{i}U\sigma_*\ot\proj{\alpha}U^\dag \mathbb{I}\ot \proj{i})\\
&=& D_{\max}(U\rho\ot\proj{\alpha}U^\dag||U\sigma_*\ot\proj{\alpha}U^\dag)\\
&=&D_{\max}(\rho\ot\proj{\alpha}||\sigma_*\ot\proj{\alpha})\\
&=&D_{\max}(\rho||\sigma_*)\\
&=&C_{\max}(\rho),
\end{eqnarray*}
where the first inequality comes from the proof of Theorem 1 in \cite{Datta2009IEEE}, the second inequality comes from the fact that there exists an extended Hilbert space $\cH_E$, a pure $\ket{\alpha}\in\cH_E$ and a
global unitary $U$ on $\cH\ot\cH_E$ such that $\Ptr{E}{\mathbb{I}\ot \proj{i}U\rho\ot\proj{\alpha}U^\dag \mathbb{I}\ot \proj{i}}=K_i\rho K^\dag_i$ \cite{Vedral1998}, the third
inequality comes from the fact that $D_{\max}$ is monotone under partial trace \cite{Datta2009IEEE}, the last inequality
comes from the fact that for any set of mutually orthogonal projectors $\set{P_k}$,
$D_{\max}(\sum_kP_k\rho_1 P_k||\sum_kP_k\rho_2 P_k)=\sum_kD_{\max}(P_k\rho_1 P_k||P_k\rho_2 P_k)$ \cite{Datta2009IEEE}
 and the first equality comes
from the fact that $D_{\max}$ is invariant under unitary operation and
$D_{\max}(\rho_1\ot P||\rho_2\ot P)=D_{\max}(\rho_1||\rho_2)$ for any projector $P$ \cite{Datta2009IEEE}.
Besides, since $C_{\max}(\tilde{\rho}_i)=\min_{\tau\in\cI}D_{\max}(\tilde{\rho}_i||\tau)\leq D_{\max}(\tilde{\rho}_i||\tilde{\sigma}_i)$, we have $\sum_i p_iC_{\max}(\tilde{\rho}_i)\leq C_{\max}(\rho)$.

\section{Coherence measure induced from $C_{\max}$}\label{apen:max_con}
Here we introduce a proper coherence measure from
$C_{\max}$ by the method of convex roof and prove that it satisfies all the
conditions (including convexity) a coherence measure need to fulfil.
We define the convex roof of $C_{\max}$ as
follows
\begin{eqnarray}
\tilde{C}_{\max}(\rho)
=\min_{\rho=\sum\lambda_i\proj{\psi_i}}
\sum_i\lambda_iC_{\max}(\psi_i),
\end{eqnarray}
where the minimum is taken over all the
pure state decompositions of state $\rho$.
Due to the definition of $\tilde{C}_{\max}$
and the properties of $C_{\max}$, the
positivity and convexity of $\tilde{C}_{\max}$
are obvious.
We only need to prove that it is nonincreasing on average under IO.

\begin{prop}
Given a quantum state $\rho\in\cD(\cH)$,
for any incoherent operation $\Phi(\cdot)=\sum_{\mu} K_{\mu}(\cdot)K^\dag_{\mu}$
with $K_{\mu}\cI K^\dag_{\mu}\subset\cI$,
\begin{eqnarray}
\sum_{\mu} p_{\mu}\tilde{C}_{\max}(\tilde{\rho}_{\mu})\leq \tilde{C}_{\max}(\rho),
\end{eqnarray}
where $p_{\mu}=\Tr{K_{\mu}\rho K^\dag_{\mu}}$ and $\tilde{\rho}_{\mu}=K_{\mu}\rho K^\dag_{\mu}/p_{\mu}$.
\end{prop}
\begin{proof}
Due to the definition of $\tilde{C}_{\max}(\rho)$, there exists a pure state decomposition of state $\rho=\sum_j\lambda_j\proj{\psi_j}$
such that $\tilde{C}_{\max}(\rho)=\sum_j\lambda_jC_{\max}(\psi_j)$. Then
\begin{eqnarray*}
\tilde{\rho}_{\mu}&=&\frac{K_{\mu}\rho K^\dag_{\mu}}{p_{\mu}}\\
&=&\sum_{j}\frac{\lambda_j}{p_u}K_{\mu}\proj{\psi_j}K^\dag_{\mu}\\
&=&\sum_{j}\frac{\lambda_jq^{(\mu)}_j}{p_{\mu}}\proj{\phi^{(\mu)}_j},
\end{eqnarray*}
where $\ket{\phi^{(\mu)}_j}=K_{\mu}\ket{\psi_j}/\sqrt{q^{(\mu)}_j}$ and
 $q^{(\mu)}_j=\Tr{K_{\mu}\proj{\psi_j}K^\dag_{\mu}}$.
 Thus, $\tilde{C}_{\max}(\tilde{\rho}_{\mu})
 \leq\sum_j\frac{\lambda_jq^{(\mu)}_j}{p_{\mu}}C_{\max}(\phi^{(\mu)}_j)$
 and
 \begin{eqnarray*}
 \sum_{\mu}p_{\mu}\tilde{C}_{\max}(\tilde{\rho}_{\mu})
 &\leq&\sum_{j,\mu}\lambda_jq^{(\mu)}_jC_{\max}(\phi^{(\mu)}_j)\\
 &=&\sum_j\lambda_j\sum_{\mu}q^{(\mu)}_jC_{\max}(\phi^{(\mu)}_j)\\
 &=&\sum_j\lambda_j\sum_{\mu}q^{(\mu)}_j\log(1+C_{l_1}(\phi^{(\mu)}_j))\\
 &\leq& \sum_j\lambda_j\log(1+\sum_{\mu}q^{(\mu)}_j C_{l_1}(\phi^{(\mu)}_j))\\
 &\leq&\sum_j\lambda_j\log(1+C_{l_1}(\psi_j))\\
 &=&\sum_j\lambda_jC_{\max}(\psi_j)\\
 &=&\tilde{C}_{\max}(\rho),
  \end{eqnarray*}
where the third line comes from the fact that for pure state $\psi$, $C_{\max}(\psi)=\log(1+C_{l_1}(\psi))$, the
forth line comes from the concavity of logarithm and the fifth lines comes from the fact that monotonicity of $C_{l_1}$ under
IO as $\Phi(\psi_j)=\sum_{\mu}K_{\mu}\proj{\psi_j}K^\dag_{\mu}=\sum_{\mu}q^{(\mu)}_j\proj{\phi^{(\mu)}_j}$.
\end{proof}

\section{The operational interpretation of $C_{\max}$}\label{apen:max_fid}
To prove the results, we need some preparation.
First of all,
Semidefinite programming (SDP) is a powerful tool in this work---which is a generalization of linear programming problems \cite{semidefinite1996}. A SDP over $\mathcal{X}=\mathbb{C}^N$ and $\mathcal{Y}=\mathbb{C}^M$ is a triple ($\Phi$, $C$, $D$), where $\Phi$ is a Hermiticity-preserving map from $\mathcal{L(X)}$ (linear operators on $\mathcal{X}$) to $\mathcal{L(Y)}$ (linear operators on $\mathcal{Y}$), $C\in$ Herm($\mathcal{X}$) (Hermitian operators over $\mathcal{X}$), and $D\in$ Herm($\mathcal{Y})$ (Hermitian operators over $\mathcal{Y}$).
There is a pair of optimization problems associated with every SDP ($\Phi$, $C$, $D$), known as the primal and the dual problems. The standard form of an SDP (that is typically followed for general conic programming) is \cite{watrous2009}
\begin{equation}
\begin{matrix}
\textrm{\underline{Primal problem}} \textrm{                          } &  \textrm{\underline{Dual problem}} \vspace{2mm}\\
\textrm{minimize: } \langle C,X \rangle, \textrm{                      } & \textrm{maximize: } \langle D,Y \rangle, \\
\textrm{subject to: } \Phi(X) \geq D, \textrm{                      } & \textrm{subject to: } \Phi^*(Y) \geq C, \\
X \in Pos(\mathcal{X}). & Y \in Pos(\mathcal{Y}).
\end{matrix}
\end{equation}
SDP forms have interesting and ubiquitous applications in quantum information theory. For example, it was recently shown by Brandao \textit{et. al} \cite{brandao2016quantum} that there exists a quantum algorithm for solving SDPs that gives an unconditional square-root speedup over any existing classical method.\\

\begin{lem}\label{lem:1}
Given a quantum state $\rho\in\cD(\cH)$,
\begin{eqnarray}
\min_{\substack{
\sigma\geq0\\
\Delta(\sigma)\geq \rho
}}\Tr{\sigma}=\max_{\substack{
\tau\geq0\\
\Delta(\tau)=\mathbb{I}
}}\Tr{\rho\tau}.
\end{eqnarray}
\end{lem}
\begin{proof}
First, we prove that
\begin{eqnarray}
\max_{\substack{
\tau\geq0\\
\Delta(\tau)=\mathbb{I}
}}\Tr{\rho\tau}
=\max_{\substack{
\tau\geq0\\
\Delta(\tau)\leq \mathbb{I}
}}\Tr{\rho\tau}.
\end{eqnarray}
For any positive operator $\tau\geq0$ with $\Delta(\tau)\leq \mathbb{I}$, define $\tau'=\tau+\mathbb{I}-\Delta(\tau)\geq0$, then  $\Delta(\tau)=\mathbb{I}$ and
$\Tr{\rho\tau'}\geq \Tr{\rho\tau}$. Thus we obtain the above equation.

Now, we prove that
\begin{eqnarray}\label{eq:SDP1}
\min_{\substack{
\sigma\geq0\\
\Delta(\sigma)\geq \rho
}}\Tr{\sigma}=\max_{\substack{
\tau\geq0\\
\Delta(\tau)\leq \mathbb{I}
}}\Tr{\rho\tau}.
\end{eqnarray}
The left side of equation \eqref{eq:SDP1} can be expressed as the following semidefinite
programming (SDP)
\begin{eqnarray*}
\min \Tr{B\sigma},\\
\text{s.t.}~~ \Lambda(\sigma)\geq C,\\
\sigma\geq 0,
\end{eqnarray*}
where $B=\mathbb{I}$, $C=\rho$ and $\Lambda=\Delta$. Then the dual SDP is given by
\begin{eqnarray*}
\max \Tr{C\tau},\\
\text{s.t.}~~ \Lambda^\dag(\tau)\leq B,\\
\tau\geq 0.
\end{eqnarray*}
That is,
\begin{eqnarray*}
\max \Tr{\rho\tau},\\
\text{s.t.}~~ \Delta(\tau)\leq \mathbb{I},\\
\tau\geq 0.
\end{eqnarray*}
Note that the dual is strictly feasible as we only need to choose $\sigma=2\lambda_{\max}(\rho) \mathbb{I}$, where $\lambda_{\max}(\rho)$ is the
 maximum eigenvalue of $\rho$. Thus, strong duality holds, and the equation \eqref{eq:SDP1} is proved.

\end{proof}

\begin{lem}\label{lem:2}
For maximally coherent state $\ket{\Psi_+}=\frac{1}{\sqrt{d}}\sum^d_{i=1}\ket{i}$, we have the following facts,

(i) For any $\mathcal{E}\in DIO$, $\tau=d\mathcal{E}^\dag(\proj{\Psi_+})$ satisfies $\tau\geq 0$ and $\Delta(\tau)=\mathbb{I}$.

(ii) For any operator $\tau\geq 0$ with $\Delta(\tau)=\mathbb{I}$, there exists a quantum operation $\mathcal{E}\in DIO$ such that $\tau=d\mathcal{E}^\dag(\proj{\Psi_+})$.

(iii) For any $\mathcal{E}\in IO$, $\tau=d\mathcal{E}^\dag(\proj{\Psi_+})$ satisfies $\tau\geq 0$ and $\Delta(\tau)=\mathbb{I}$.

(iv) For any operator $\tau\geq 0$ with $\Delta(\tau)=\mathbb{I}$, there exists a quantum operation $\mathcal{E}\in IO$ such that $\tau=d\mathcal{E}^\dag(\proj{\Psi_+})$.

(v) For any $\mathcal{E}\in SIO$, $\tau=d\mathcal{E}^\dag(\proj{\Psi_+})$ satisfies $\tau\geq 0$ and $\Delta(\tau)=\mathbb{I}$.

(vi) For any operator $\tau\geq 0$ with $\Delta(\tau)=\mathbb{I}$, there exists a quantum operation $\mathcal{E}\in SIO$ such that $\tau=d\mathcal{E}^\dag(\proj{\Psi_+})$.

\end{lem}
\begin{proof}
(i) Since $\mathcal{E}$ is a CPTP map, $\mathcal{E}^\dag$ is unital. Besides, as $\mathcal{E}\in DIO$, $[\mathcal{E},\Delta]=0$ implies that $[\mathcal{E}^\dag, \Delta]=0$. Thus $\Delta(\tau)=d\mathcal{E}^\dag(\Delta(\proj{\Psi_+}))=\mathcal{E}^\dag(\mathbb{I})=\mathbb{I}$.

(ii) For any positive operator $\tau\geq 0$ with $\Delta(\tau)=\mathbb{I}$,
$\Tr{\tau}=d$, thus $\tau=d\hat{\tau}$ with $\hat{\tau}\in\cD(\cH)$ and $\Delta(\hat{\tau})=\frac{1}{d}\mathbb{I}$.
Consider the spectral decomposition of  $\hat{\tau}=\sum^d_{i=1}\lambda_i\proj{\psi_i}$ with $\sum^d_{i=1}\lambda_i=1, \lambda_i\geq 0$ for any $i\in \set{1,.., d}$.
Besides,  for any $i\in \set{1,..., d}$, $\ket{\psi_i}$ can be written as
$\ket{\psi_i}=\sum^d_{j=1}c^{(i)}_j\ket{j}$ with $\sum^d_{j=1}|c^{(i)}_j|^2=1$.
Let us define $K^{(i)}_n=\sum^d_{j=1}c^{(i)}_j\proj{j}$ for any $n\in \set{1,...,d}$, then $K^{(i)}_n\ket{\Psi_+}=\frac{1}{\sqrt{d}}\ket{\psi_i}$ and
$\sum^d_{n=1}K^{(i)}_n\proj{\Psi_+}K^{(i)\dag}_n=\proj{\psi_i}$.
Let $M_{i,n}=\sqrt{\lambda_i}K^{(i)\dag}_n$,
then
\begin{eqnarray*}
\sum_{i,n}M^\dag_{i,n}M_{i,n}&=&\sum_{i,n}\lambda_iK^{(i)}_nK^{(i)\dag}_n\\
&=&d\sum^d_{i=1}\lambda_iK^{(i)}_1K^{(i)\dag}_1\\
&=&d\sum^d_{i=1}\lambda_i\sum^d_{j=1}|c^{(i)}_j|^2\proj{j}\\
&=&d\sum^d_{j=1}\sum^d_{i=1}\lambda_i|c^{(i)}_j|^2\proj{j}\\
&=&d\sum^d_{j=1}\frac{1}{d}\proj{j}=\mathbb{I},
\end{eqnarray*}
 where $\sum^d_{i=1}\lambda_i|c^{(i)}_j|^2=\sum_i\lambda_i|\iinner{\psi_i}{j}|^2
 =\Innerm{j}{\hat{\tau}}{j}=\frac{1}{d}$.
Then $\mathcal{E}(\cdot)=\sum_{i,n}M_{i,n}(\cdot) M^\dag_{i,n}$ is a CPTP map. Since $M_{i,n}$ is diagonal,
the quantum operation $\mathcal{E}(\cdot)=\sum_{i,n}M_{i,n}(\cdot) M^\dag_{i,n}$ is a DIO.
Moreover, $\mathcal{E}^\dag(\proj{\Psi_+})=\sum_{i,n}M^\dag_{i,n}\proj{\Psi_+}M_{i,n}
=\sum_{i,n}\lambda_iK^{(i)}_n\proj{\Psi_+}K^{(i)\dag}_n=\sum_i\lambda_i\proj{\psi_i}=\hat{\tau}$.

(iii) If $\cE$ is an incoherent operation, then there exists a set of Kraus operators $\set{K_{\mu}}$ such that
$\cE(\cdot)=\sum_{\mu}K_{\mu}(\cdot) K^\dag_{\mu}$ and $K_{\mu}\cI K^{\dag}_{\mu}\in\cI$. Thus
\begin{eqnarray*}
d\Innerm{i}{\cE^{\dag}(\proj{\Psi_+})}{i}
&=&d\sum_{\mu}\Innerm{i}{ K^\dag_{\mu}\proj{\Psi_+}K_{\mu}}{i}\\
&=&\sum_{\mu}\sum_{m,n}\Innerm{i}{K^{\dag}_{\mu}}{m}\!\Innerm{n}{K_{\mu}}{i}\\
&=&\sum_{\mu}\sum_{m}\Innerm{i}{K^{\dag}_{\mu}}{m}\!\Innerm{m}{K_{\mu}}{i}\\
&=&\sum_{\mu}\Innerm{i}{K^{\dag}_{\mu}K_{\mu}}{i}
=1,
\end{eqnarray*}
where the third line comes from the fact that
for any $K_{\mu}$, there exists at most one nonzero term in each  column which implies
that $\bra{i}K^\dag_{\mu}\ket{m}\bra{n}K_{\mu}\ket{i}\neq 0$ only if $m=n$, and the forth line
comes from the fact that $\sum_{\mu}K^\dag_{\mu}K_{\mu}=\mathbb{I}$. Therefore,
$\Delta(d\cE^\dag(\proj{\Psi_+}))=\mathbb{I}$.

(iv) This is obvious, as the DIO $\cE$ given in (ii) is also an incoherent operation.

(v) This is obvious as $SIO\subset DIO$.

(vi) This is obvious as the DIO $\cE$ given in (ii) also belongs to $SIO$.

\end{proof}

\begin{lem}\label{lem:3}
Given a quantum state $\rho\in\cD(\cH)$, one has
\begin{eqnarray*}
\max_{
\mathcal{E}\in DIO
}F(\mathcal{E}(\rho),\proj{\Psi_+})^2
=
\max_{\substack{
\mathcal{E}\in DIO\\
\ket{\Psi}\in\cM
}}F(\mathcal{E}(\rho),\proj{\Psi})^2.
\end{eqnarray*}
where $\ket{\Psi_+}=\frac{1}{\sqrt{d}}\sum^d_{i=1}\ket{i}$ and
$\cM$ is the set of maximally coherent states.
\end{lem}
\begin{proof}
Due to \cite{Peng2015}, every maximally coherent can be expressed as
$\ket{\Psi}=\frac{1}{\sqrt{d}}\sum^d_{j=1}e^{i\theta_j}\ket{j}$, that is,
$\ket{\Psi}=U_{\Psi}\ket{\Psi_+}$ where $U_{\Psi}=\sum^d_{j=1}e^{i\theta_j}\proj{j}$.
Obviously, $[U_{\Psi}, \Delta]=0$, thus $U_{\Psi}\in DIO$ and
\begin{eqnarray*}
F(\mathcal{E}(\rho),\proj{\Psi})^2
&=&F(\mathcal{E}(\rho), U_{\Psi}\proj{\Psi_+}U^\dag_{\Psi})^2\\
&=&F(U^\dag_{\Psi}\mathcal{E}(\rho)U_{\Psi},\proj{\Psi})^2\\
&=&F(\mathcal{E}'(\rho),\proj{\Psi_+})^2,
\end{eqnarray*}
where $\mathcal{E}'(\cdot)=U^\dag_{\Psi}\mathcal{E}(\cdot)U_{\Psi}\in DIO$ as
$\mathcal{E}, U_{\Psi}\in DIO$.
\end{proof}

After these preparation, we begin to prove Theorem 1.

\begin{mproof}[Proof of Theorem 1]
If $\mathcal{E}$ belongs to
DIO, that is, we need to prove
\begin{eqnarray}\label{eq:DIO}
2^{C_{\max}(\rho)}=d\max_{\substack{
\mathcal{E}\in DIO\\
\ket{\Psi}\in\cM
}}F(\mathcal{E}(\rho),\proj{\Psi})^2,
\end{eqnarray}
where
$\cM$ is the set of maximally coherent states.
In view of Lemma \ref{lem:3},  we only need to prove
\begin{eqnarray}
2^{C_{\max}(\rho)}=d\max_{
\mathcal{E}\in DIO
}F(\mathcal{E}(\rho),\proj{\Psi_+})^2,
\end{eqnarray}
where $\ket{\Psi_+}=\frac{1}{\sqrt{d}}\sum^d_{i=1}\ket{i}$.

First of all,
\begin{eqnarray*}
2^{C_{\max}(\rho)}
&=&\min_{\sigma\in\cI} \min\set{\lambda|\rho\leq \lambda\sigma}\\
&=&\min_{\sigma\geq0}\set{\Tr{\sigma}|\rho\leq \Delta(\sigma)}\\
&=&\min_{\substack{
\sigma\geq0\\
\Delta(\sigma)\geq \rho
}}\Tr{\sigma}.
\end{eqnarray*}

Second,
\begin{eqnarray*}
dF(\mathcal{E}(\rho),\proj{\Psi_+})^2
&=&d\Tr{\mathcal{E}(\rho)\proj{\Psi_+}}\\
&=&d\Tr{\rho\mathcal{E}^\dag(\proj{\Psi_+})}\\
&=&\Tr{\rho\tau},
\end{eqnarray*}
where $\tau=d\mathcal{E}^\dag(\proj{\Psi_+}$. According to Lemma \ref{lem:2}, there is
one to one correspondence between $DIO$ and the set $\set{\tau\geq0|\Delta(\tau)=\mathbb{I}}$. Thus we have

\begin{eqnarray}
d\max_{
\mathcal{E}\in DIO
}F(\mathcal{E}(\rho),\proj{\Psi_+})^2
=\max_{\substack{
\tau\geq0\\
\Delta(\tau)=\mathbb{I}
}}\Tr{\rho\tau}.
\end{eqnarray}
Finally, according to Lemma \ref{lem:1}, we get the desired result \eqref{eq:DIO}.
Similarly, we can prove
the case where $\mathcal{E}$ belongs to
either IO or SIO based on Lemma \ref{lem:2}.
\end{mproof}

\section{Subchannel discrimination in dephasing covariant instrument}\label{apen:max_sub_dis}

\begin{mproof}[Proof of Theorem 2]
First, we consider the case where instrument $\mathfrak{I}$ is dephasing-covariant instrument  $\mathfrak{I}^D$. Due to the definition of $C_{\max}(\rho)$, there exists an incoherent state $\sigma$
such that $\rho\leq 2^{C_{\max}(\rho)}\sigma$. Thus, for any
dephasing-covariant instrument $\mathfrak{I}^D$ and
POVM $\set{M_b}_b$,
\begin{eqnarray*}
p_{\text{succ}}(\mathfrak{I}^D,\set{M_b}_b,\rho)\leq 2^{C_{\max}(\rho)} p_{\text{succ}}(\mathfrak{I}^D, \set{M_b}_b, \sigma),
\end{eqnarray*}
which implies that
\begin{eqnarray}\label{eq:dis_eql}
p_{\text{succ}}(\mathfrak{I}^D,\rho)\leq 2^{C_{\max}(\rho)} p^{ICO}_{\text{succ}}(\mathfrak{I}^D).
\end{eqnarray}

Next, we prove that there exists a dephasing-covariant instrument
$\mathfrak{I}^D$ such that the equality in \eqref{eq:dis_eql} holds.
In view of Theorem 1, there exists a DIO $\cE$
such that
\begin{eqnarray}
2^{C_{\max}(\rho)}=d\Tr{\cE(\rho)\proj{\Psi_+}},
\end{eqnarray}
where $\ket{\Psi_+}=\frac{1}{\sqrt{d}}\sum^d_{j=1}\ket{j}$.
Let us consider the following diagonal unitaries
\begin{eqnarray}
U_{k}=\sum^d_{j=1}e^{i\frac{jk}{d}2\pi}\proj{j},k\in\set{1,..,d}.
\end{eqnarray}
The set $\set{U_k\!\ket{\Psi_+}}^d_{k=1}$  forms a basis of the Hilbert space
and $\sum^d_{k=1}U_k\proj{\Psi_+}U^\dag_k=\mathbb{I}$.
Let us define subchannels $\set{\cE_k}_k$ as $\cE_k(\rho)=\frac{1}{d}U_k\cE(\rho) U^\dag_k$. Then
the channel $\widetilde{\cE}=\sum^d_{k=1}\cE_k$ is a DIO. That is, the instrument $\widetilde{\mathfrak{I}}^D=\set{\cE_k}_k$ is a
dephasing-covariant instrument.

For any POVM $\set{M_k}_k$  and any incoherent state $\sigma$, the probability of success is
\begin{eqnarray*}
p_{\text{succ}}(\widetilde{\mathfrak{I}}^D, \set{M_k}_k,\sigma)
&=&\sum_k\Tr{\cE_k(\sigma)M_k}\\
&=&\frac{1}{d}\Tr{\cE(\sigma)\sum_kU^\dag_kM_kU_k}.
\end{eqnarray*}
Since $\set{M_k}_k$ is a POVM, then
$\sum_kM_k=\mathbb{I}$. As $\set{U_k}_k$
are all diagonal unitaries , we have
\begin{eqnarray*}
\Delta(\sum_kU^\dag_kM_kU_k)
&=&\sum_k U^\dag_k\Delta(M_k)U_k \\
&=&\sum_k\Delta(M_k)\\
&=&\Delta(\sum_k M_k)\\
&=&\Delta(\mathbb{I})=\mathbb{I}.
\end{eqnarray*}
Thus,
\begin{eqnarray*}
p_{\text{succ}}(\widetilde{\mathfrak{I}}^D, \set{M_k}_k,\sigma)
&=&\frac{1}{d}\Tr{\cE(\sigma)\sum_kU^\dag_kM_kU_k}\\
&=&\frac{1}{d}\Tr{\Delta(\cE(\sigma))\sum_kU^\dag_kM_kU_k}\\
&=&\frac{1}{d}\Tr{\cE(\sigma)\Delta(\sum_kU^\dag_kM_kU_k)}\\
&=&\frac{1}{d}\Tr{\cE(\sigma)}
=\frac{1}{d},
\end{eqnarray*}
where the second equality comes from the fact that $\cE(\sigma)\in\cI$ for any incoherent state $\sigma$,  and the second last equality
comes from that fact that
$\Delta(\sum_kU^\dag_kM_kU_k)=\mathbb{I}$.
That is,
\begin{eqnarray}
p^{ICO}_{\text{succ}}(\widetilde{\mathfrak{I}}^D)=\frac{1}{d}.
\end{eqnarray}

Besides, taking the POVM  $\set{N_k}_k$ with $N_k=U_k\proj{\Psi_+}U^\dag_k$,
one has $\Tr{\cE_k(\rho)N_k}=\frac{1}{d}\Tr{\cE(\rho)\proj{\Psi_+}}$ and
\begin{eqnarray*}
p_{\text{succ}}(\widetilde{\mathfrak{I}}^D, \set{N_k}_k,\rho)
&=&\sum_k\Tr{\cE_k(\rho)N_k}\\
&=&\sum_k\frac{1}{d}\Tr{\cE(\rho)\proj{\Psi_+}}\\
&=&\Tr{\cE(\rho)\proj{\Psi_+}}\\
&=&\frac{2^{C_{\max}(\rho)}}{d}\\
&=&2^{C_{\max}(\rho)}p^{ICO}_{\text{succ}}(\widetilde{\mathfrak{I}}^D).
\end{eqnarray*}
Thus, for this depasing-covariant instrument $\widetilde{\mathfrak{I}}^D=\set{\cE_k}_k$,
\begin{eqnarray}
\frac{p_{\text{succ}}(\widetilde{\mathfrak{I}}^D,\rho)}{p^{ICO}_{\text{succ}}(\widetilde{\mathfrak{I}}^D)}
\geq2^{C_{\max}(\rho)}.
\end{eqnarray}

Finally it is easy to see that the above proof is also true for $\mathfrak{I}$  is
$\mathfrak{I}^I$ or $\mathfrak{I}^S$.

\end{mproof}

Note that the phasing discrimination game studied in \cite{Napoli2016} is just a special
case of the subchannel discrimination in the dephasing-covariant instruments.
In the phasing discrimination game, the phase $\phi_k$
is encoded into a diagonal unitary $U_{\phi_k}=\sum_j e^{ij\phi_k}\proj{j}$.
Thus the discrimination of a collection of phase $\set{\phi_k}$ with
a prior probability distribution $\set{p_k}$ is equivalent to the
discrimination of
the set of subchannel $\set{\cE_k}_k$, where $\cE_k=p_k\mathbf{U}_{k}$ and
$\mathbf{U}_k(\cdot)=U_{\phi_k}(\cdot) U^\dag_{\phi_k}$.

\section{$C^{\epsilon}_{\max}$ as a lower bound of one-shot coherence cost }

 The $\epsilon$-smoothed max-relative entropy of coherence of a quantum state $\rho$ is defined by,
\begin{eqnarray}
C^{\epsilon}_{\max}(\rho):=\min_{\rho'\in B_{\epsilon}(\rho)}C_{\max}(\rho'),
\end{eqnarray}
where $B_\epsilon(\rho):=\set{\rho'\geq0:\norm{\rho'-\rho}_1\leq \epsilon, \Tr{\rho'}\leq\Tr{\rho}}$.
Then
\begin{eqnarray*}
C^{\epsilon}_{\max}(\rho)&=&\min_{\rho'\in B_{\epsilon}(\rho)}\min_{\sigma\in \cI}D_{\max}(\rho'||\sigma)\\
&=&\min_{\sigma\in\cI} D^{\epsilon}_{\max}(\rho||\sigma),
\end{eqnarray*}
where $D^{\epsilon}_{\max}(\rho||\sigma)$  is the smooth max-relative entropy \cite{Datta2009IEEE,Datta2009,Brandao2011} and defined as
\begin{eqnarray*}
D^\epsilon_{\max}(\rho||\sigma)=\inf_{\rho'\in B_{\epsilon}(\rho)}D_{\min}(\rho'||\sigma).
\end{eqnarray*}

\begin{mproof}[Proof of Equation (5)]
Suppose $\mathcal{E}$ is MIO  such that
$F(\mathcal{E}(\proj{\Psi^M_+}),\rho)^2\geq 1-\epsilon$ and
$ C^{(1),\epsilon}_{C,MIO}(\rho)=\log M$.
Since $F(\rho,\sigma)^2\leq 1-\frac{1}{4}\norm{\rho-\sigma}^2_1$ \cite{Nielsen10}, then
$\norm{\mathcal{E}(\proj{\Psi^M_+})-\rho}_1\leq2\sqrt{\epsilon}$. Thus
$\mathcal{E}(\proj{\Psi^M_+})\in B_{\epsilon'}(\rho)$, where $\epsilon'=2\sqrt{\epsilon}$.
As $C_{\max}$
 is monotone under MIO, we have
$C^{\epsilon'}_{\max}(\rho)
\leq C_{\max}(\mathcal{E}(\proj{\Psi^M_+}))\leq C_{\max}(\proj{\Psi^M_+})=\log M= C^{(1),\epsilon}_{C,MIO}(\rho)$.
\end{mproof}

\section{Equivalence between $C_{\max}$ and $C_r$ in asymptotic case}\label{apen:eq_asy}

We introduce several lemmas first to prove the
result. For any self-adjoint operator $Q$ on a finite-dimensional
Hilbert space, $Q$ has the spectral decomposition as
$Q=\sum_i\lambda_iP_i$, where $P_i$ is the orthogonal projector onto
the eigenspace of $Q$. Then we define the positive operator
$\set{Q\geq 0}=\sum_{\lambda_i\geq 0}P_i$, and  $\set{Q>0}$, $\set{Q\leq 0}$, $\set{Q<0}$ are defined in a similar way.  Moreover, for any two operators
$Q_1$ and $Q_2$, $\set{Q_1\geq Q_2}$ is defined as $\set{Q_1-Q_2\geq 0}$.

\begin{lem}\cite{Datta2009}\label{lem:4}
Given two quantum states $\rho,\sigma\in\cD(\cH)$, then
\begin{eqnarray}
D^{\epsilon}_{\max}(\rho||\sigma)\leq\lambda
\end{eqnarray}
for any $\lambda\in\real$ and $\epsilon=\sqrt{8\Tr{\set{\rho>2^\lambda\sigma}\rho}}$.
\end{lem}

Note that in \cite{Datta2009}, the above lemma is proved for bipartite states. However,
this lemma also holds for any state.

\begin{lem}
(\text{Fannes-Audenaert Inequality \cite{Audenaert2007}})
For any two quantum states $\rho$ and $\sigma$ with $\epsilon=\frac{1}{2}\norm{\rho-\sigma}_1$, the following inequality holds:
\begin{eqnarray}\label{ineq:FA}
|S(\rho)-S(\sigma)|\leq \epsilon\log(d-1)+H_2(\epsilon),
\end{eqnarray}
where $d$ is the dimension of the system and $H_2(\epsilon)
=-\epsilon\log \epsilon-(1-\epsilon)\log(1-\epsilon)$ is the binary
Shannon entropy.
\end{lem}

Based on these lemmas, we can prove the equivalence between $C_{\max}$ and $C_r$ in asymptotic limit.
\begin{mproof}[Proof of Equation (6)]
First, we prove that
\begin{eqnarray}
C_r(\rho)\leq\lim_{\epsilon\to 0}\lim_{n\to \infty}
\frac{1}{n}C^{\epsilon}_{\max}(\rho^{\ot n}).
\end{eqnarray}

Since
\begin{eqnarray*}
C^{\epsilon}_{\max}(\rho^{\ot n})
&=&\min_{\sigma_n\in\cI_n}D^{\epsilon}_{\max}(\rho^{\ot n}||\sigma_n)\\
&=&D_{\max}(\rho_{n,\epsilon}||\tilde{\sigma}_n),
\end{eqnarray*}
where $\rho_{n,\epsilon}\in B_{\epsilon}(\rho^{\ot n})$,
$\tilde{\sigma}_n\in \cI_n$ and $\cI_n$ is the set of incoherent states of
$\cH^{\ot n}$.

Due to the definition of $\cD_{\max}$, we have
\begin{eqnarray}
\rho_{n,\epsilon}\leq 2^{C^{\epsilon}_{\max}(\rho^{\ot n})}\tilde{\sigma}_n.
\end{eqnarray}
Then
\begin{eqnarray*}
C_r(\rho_{n,\epsilon})
&\leq& S(\rho_{n,\epsilon}||\tilde{\sigma}_n)\\
&=&\Tr{\rho_{n,\epsilon}\log \rho_{n,\epsilon}}-\Tr{\rho_{n,\epsilon}\log \tilde{\sigma}_n}\\
&\leq&\Tr{\rho_{n,\epsilon}(C^{\epsilon}_{max}(\rho^{\ot n})+\log\tilde{\sigma}_n)}-\Tr{\rho_{n,\epsilon}\log \tilde{\sigma}_n}\\
&\leq& C^{\epsilon}_{\max}(\rho^{\ot n}),
\end{eqnarray*}
where $\Tr{\rho_{n,\epsilon}}\leq \Tr{\rho^{\ot n}}=1$. Besides, as $\norm{\rho_{n,\epsilon}-\rho^{\ot n}}\leq \epsilon$,
due to the Fannes-Audenaert Inequality \eqref{ineq:FA}, we have
\begin{eqnarray*}
C_r(\rho_{n,\epsilon})
&\geq& C_r(\rho^{\ot n})-\epsilon\log(d-1)-H_2(\epsilon)\\
&=&nC_r(\rho)-\epsilon\log(d-1)-H_2(\epsilon).
\end{eqnarray*}
Thus,
\begin{eqnarray}
\lim_{\epsilon\to 0}\lim_{n\to\infty}\frac{1}{n}C^{\epsilon}_{\max}(\rho^{\ot n})
\geq C_r(\rho).
\end{eqnarray}

Next, we prove that
\begin{eqnarray*}
C_r(\rho)\geq \lim_{\epsilon\to 0}\lim_{n\to \infty}
\frac{1}{n}C^{\epsilon}_{\max}(\rho^{\ot n}).
\end{eqnarray*}
Consider the sequence $\hat{\rho}=\set{\rho^{\ot n}}^{\infty}_{n=1}$ and
$\hat{\sigma}_I=\set{\sigma^{\ot n}_I}^\infty_{n=1}$, where $\sigma_I\in\cI$ such that $C_r(\rho)=S(\rho||\sigma_I)=\min_{\sigma\in \cI} S(\rho||\sigma)$.
Denote
\begin{eqnarray*}
\overline{D}(\hat{\rho}||\hat{\sigma})
:=\inf\set{\gamma:\lim_{n\to \infty}\sup\Tr{\set{\rho^{\ot n}\geq 2^{n\gamma}\sigma^{\ot n}_I}\rho^{\ot n}}=0   }.
\end{eqnarray*}
Due to the Quantum Stein's Lemma \cite{Ogawa2000,Datta2009},
\begin{eqnarray*}
\overline{D}(\hat{\rho}||\hat{\sigma})
=S(\rho||\sigma_I)=C_r(\rho).
\end{eqnarray*}

For any $\delta>0$, let $\lambda=\overline{D}(\hat{\rho}||\hat{\sigma})+\delta=S(\rho||\sigma_I)=C_r(\rho)+\delta$.
Due to the definition of the quantity $\overline{D}(\hat{\rho}||\hat{\sigma})$, we have
\begin{eqnarray*}
\lim_{n\to \infty}\sup\Tr{\set{\rho^{\ot n}\geq 2^{n\lambda}\sigma^{\ot n}_I}\rho^{\ot n}}=0.
\end{eqnarray*}
Then  for any $\epsilon>0$, there exists an integer $N_0$ such that for any $n\geq N_0$, $\Tr{\set{\rho^{\ot n}\geq 2^{n\lambda}\sigma^{\ot n}_I}\rho^{\ot n}}<\frac{\epsilon^2}{8}$.
According to Lemma \ref{lem:4}, we have
\begin{eqnarray}
D^{\epsilon}_{\max}(\rho^{\ot n}||\sigma^{\ot n}_I)\leq n\lambda =nC_r(\rho)+n\delta.
\end{eqnarray}
for $n\geq N_0$. Hence,
\begin{eqnarray*}
C^{\epsilon}_{\max}(\rho^{\ot n})\leq nC_r(\rho)+n\delta.
\end{eqnarray*}
Therefore
\begin{eqnarray*}
\lim_{\epsilon\to 0}\lim_{n\to\infty}\frac{1}{n}
C^{\epsilon}_{\max}(\rho^{\ot n})\leq C_r(\rho)+\delta.
\end{eqnarray*}
Since $\delta$ is arbitrary,
$\lim_{\epsilon\to 0}\lim_{n\to\infty}\frac{1}{n}
C^{\epsilon}_{\max}(\rho^{\ot n})\leq C_r(\rho)$.

\end{mproof}

\section{
$C_{\min}$ may increase on average under IO
}\label{apen:min_incr}

For pure state $\ket{\psi}=\sum^d_{i=1}\psi_i\ket{i}$, we have
$C_{\min}(\psi)=-\log\max_{i}|\psi_i|^2$. According to \cite{Du2015},
if $C_{\min}$ is nonincreasing on average under IO,
it requires that $C_{\min}$ should be a concave function of its diagonal part for
pure state. However, $C_{\min}(\psi)=-\log\max_{i}|\psi_i|^2$ is convex on the
diagonal part of the pure states, hence $C_{\min}$ may increase on average
under IO.

Besides, according the definition of $C_{\min}$ for pure state $\ket{\psi}$, one has
\begin{eqnarray}
2^{-C_{\min}(\psi)}
=\max_{\sigma\in\cI}F(\psi,\sigma)^2.
\end{eqnarray}
However, this equality does not hold for any states. For any quantum state, the inequality (13) in the main context holds.

\begin{mproof}[Proof of Equation (13)]
There exists a $\sigma_*\in\cI$ such that
$F(\rho,\sigma_*)^2=\max_{\sigma\in\cI}F(\rho,\sigma)^2$. Let us consider the spectrum decomposition
of the quantum state $\rho$, $\rho=\sum_{i}\lambda_i\proj{\psi_i}$ with $\lambda_i>0$ and $\sum_i\lambda_i=1$. Then
the projector $\Pi_{\rho}$ onto the support of $\rho$  can be written as
$\Pi_{\rho}=\sum_i\proj{\psi_i}$ and $2^{-C_{\min}(\rho)}=\max_{\sigma\in\cI}\Tr{\Pi_{\rho}\sigma}$.

Besides, there exists pure state decomposition of $\sigma_*=\sum_{i}\mu_i\proj{\phi_i}$ such that
$F(\rho,\sigma_*)=\sum_i\sqrt{\lambda_i\mu_i}\iinner{\psi_i}{\phi_i}$ \cite{Hayashi2006}. Thus
\begin{eqnarray*}
F(\rho,\sigma_*)^2&=&(\sum_i\sqrt{\lambda_i\mu_i}\iinner{\psi_i}{\phi_i})^2\\
&\leq&(\sum_i\lambda_i)(\sum_i\mu_i\Tr{\proj{\psi_i}\proj{\phi_i}})\\
&\leq&\sum_i\mu_i\Tr{\Pi_{\rho}\proj{\phi_i}}\\
&=&\Tr{\Pi_{\rho}\sigma_*}\\
&\leq& 2^{-C_{\min}(\rho)},
\end{eqnarray*}
where the first inequality is due to the Cauchy-Schwarz inequality and the second inequality comes from the fact that
$\sum_i\lambda_i=1$ and $\proj{\psi_i}\leq \Pi_{\rho}$ for any i.
\end{mproof}

\section{Equivalence between $C_{\min}$ and $C_r$ in asymptotic case}\label{apen:min_asym}
For any $\epsilon>0$, the smooth min-relative entropy of coherence of a quantum state $\rho$ is defined as follows
\begin{eqnarray}
C^{\epsilon}_{\min}(\rho):=\max_{\substack{
0\leq A\leq \mathbb{I}\\
\Tr{A\rho}\geq 1-\epsilon}}
\min_{\sigma\in\cI}-\log\Tr{A\sigma},
\end{eqnarray}
where $\mathbb{I}$ denotes the
identity.
Then
\begin{eqnarray*}
C^{\epsilon}_{\min}(\rho)
&=&\max_{\substack{
0\leq A\leq \mathbb{I}\\
\Tr{A\rho}\geq 1-\epsilon}}
\min_{\sigma\in\cI}-\log\Tr{A\sigma}\\
&=&\min_{\sigma\in\cI}\max_{\substack{
0\leq A\leq \mathbb{I}\\
\Tr{A\rho}\geq 1-\epsilon}}
-\log\Tr{A\sigma}\\
&=&\min_{\sigma\in\cI}D^{\epsilon}_{\min}(\rho||\sigma),
\end{eqnarray*}
where $D^{\epsilon}_{\min}(\rho||\sigma)$ is the smooth min-relative entropy \cite{Brandao2011} and
defined as
\begin{eqnarray*}
D^\epsilon_{\min}(\rho||\sigma)=
\sup_{\substack{
0\leq A\leq \mathbb{I}\\
\Tr{A\rho}\geq 1-\epsilon}}
-\log\Tr{A\sigma}.
\end{eqnarray*}

\begin{lem}
Given a quantum state $\rho\in\cD(\cH)$, for any $\epsilon>0$,
\begin{eqnarray}
C^{\epsilon}_{\min}(\rho)\leq C^{\epsilon}_{\max}(\rho)-\log(1-2\epsilon).
\end{eqnarray}

\end{lem}
\begin{proof}
Since
\begin{eqnarray*}
C^{\epsilon}_{\max}(\rho)&=&\min_{\sigma\in\cI}D^{\epsilon}_{\max}(\rho||\sigma),\\
C^{\epsilon}_{\min}(\rho)&=&\min_{\sigma\in\cI}D^{\epsilon}_{\min}(\rho||\sigma),
\end{eqnarray*}
we only need to prove that  for any two states $\rho$ and $\sigma$,
\begin{eqnarray}
D^{\epsilon}_{\min}(\rho||\sigma)
\leq D^{\epsilon}_{\max}(\rho||\sigma)-\log(1-2\epsilon).
\end{eqnarray}

First, there exists a $\rho_{\epsilon}\in B_{\epsilon}(\rho)$ such that
$D^{\epsilon}_{\max}(\rho||\sigma)=D_{\max}(\rho_{\epsilon}||\sigma)=\log\lambda$. Hence $\lambda\sigma-\rho_{\epsilon}\geq 0$.

Second, let $0\leq A\leq \mathbb{I}$, $\Tr{A\rho}\geq 1-\epsilon$ such that
$D^{\epsilon}_{\min}(\rho||\sigma)=-\log\Tr{A\sigma}$.
Since for any two positive operators $A$ and $B$, $\Tr{AB}\geq 0$. Therefore
$\Tr{(\lambda\sigma-\rho_{\epsilon})A}\geq0$, that is,
\begin{eqnarray*}
\Tr{A\rho_{\epsilon}}
\leq\lambda\Tr{A\sigma}.
\end{eqnarray*}

Since $\norm{\rho-\rho_{\epsilon}}_1=\Tr{|\rho-\rho_{\epsilon}|}<\epsilon$ and
$\Tr{ A\rho}\geq1-\epsilon$, one gets
\begin{eqnarray*}
|\Tr{A\rho_{\epsilon}}-\Tr{A\rho}|
&\leq& \Tr{A|\rho-\rho_{\epsilon}|}\\
&\leq&\Tr{|\rho-\rho_{\epsilon}|}<\epsilon.
\end{eqnarray*}
Thus,
\begin{eqnarray*}
\Tr{A\rho_{\epsilon}}
\geq \Tr{A\rho}-\epsilon\geq 1-2\epsilon,
\end{eqnarray*}
which implies that
\begin{eqnarray*}
1-2\epsilon\leq\lambda\Tr{A\sigma}.
\end{eqnarray*}
Take logarithm on both sides of the above inequality,
we have
\begin{eqnarray*}
-\log\Tr{A\sigma}\leq\log\lambda-\log(1-2\epsilon).
\end{eqnarray*}
That is,
\begin{eqnarray}
D^{\epsilon}_{\min}(\rho||\sigma)
\leq D^{\epsilon}_{\max}(\rho||\sigma)-\log(1-2\epsilon).
\end{eqnarray}

\end{proof}

The following lemma is a kind of generalization of the Quantum Stein' Lemma \cite{Brandao2010CMP} for the special case of the incoherent state set $\cI$, as the the set of incoherent states satisfies the requirement in \cite{Brandao2010CMP}.  Note that this lemma can be generalized to any quantum resource theory which satisfies some postulates \cite{FBrandao15}
and it is called the exponential distinguishability property (EDP) (see \cite{FBrandao15}).

\begin{lem}\label{lem:stein}
Given a quantum state $\rho\in\cD(\cH)$,

(Direct part) For any $\epsilon>0$, there exists a sequence of POVMs $\set{A_n, \mathbb{I}-A_n}_n$ such that
\begin{eqnarray}
\lim_{n\to\infty}
\Tr{(\mathbb{I}-A_n)\rho^{\ot n}}=0,
\end{eqnarray}
and for every integer $n$ and incoherent state $w_n\in\cI_n$ with $\cI_n$ is the
set of incoherent states on $\cH^{\ot n}$,
\begin{eqnarray}
-\frac{\log\Tr{A_nw_n}}{n}+\epsilon
\geq C_r(\rho).
\end{eqnarray}

(Strong converse) If there exists $\epsilon>0$ and a sequence of POVMs $\set{A_n, \mathbb{I}-A_n}_n$ such that for every integer $n>0$ and $w_n\in\cI_n$,
\begin{eqnarray}
-\frac{\log\Tr{A_nw_n}}{n}-\epsilon
\geq C_r(\rho),
\end{eqnarray}
then
\begin{eqnarray}
\lim_{n\to\infty}\Tr{(\mathbb{I}-A_n)\rho^{\ot n}}=1.
\end{eqnarray}

\end{lem}

Now, we are ready to prove the equivalence between $C_{\min}$ and $C_{r}$ in asymptotic limit.
\begin{mproof}[Proof of Equation (18)]
First we prove that
\begin{eqnarray*}
\lim_{\epsilon\to 0}\lim_{n\to\infty}\frac{1}{n}C^{\epsilon}_{\min}(\rho^{\ot n})
\leq C_r(\rho).
\end{eqnarray*}
Since $C_r(\rho)=\lim_{\epsilon\to 0}\lim_{n\to\infty}\frac{1}{n}C^{\epsilon}_{\max}(\rho^{\ot n})$ and
$C^{\epsilon}_{\min}(\rho)\leq C^{\epsilon}_{\max}(\rho)-\log(1-2\epsilon)$, then
we have
\begin{eqnarray*}
\lim_{\epsilon\to 0}\lim_{n\to\infty}\frac{1}{n}C^{\epsilon}_{\min}(\rho^{\ot n})
\leq C_r(\rho).
\end{eqnarray*}

Now we prove that
\begin{eqnarray*}
\lim_{\epsilon\to 0}\lim_{n\to\infty}\frac{1}{n}C^{\epsilon}_{\min}(\rho^{\ot n})
\geq C_r(\rho)
\end{eqnarray*}
According to Lemma \ref{lem:stein},
for any $\epsilon>0$, there exists a sequence of POVMs $\set{A_n}$ such that for sufficient large integer $n$,  $\Tr{\rho^{\ot n}A_n}\geq1-\epsilon$, and thus
\begin{eqnarray*}
C^{\epsilon}_{\min}(\rho^{\ot n})
\geq \min_{\sigma\in\cI}-\log\Tr{A_n\sigma}\geq n(C_r(\rho)-\epsilon),
\end{eqnarray*}
where the last inequality comes from the direct part of Lemma \ref{lem:stein}.
Therefore,
\begin{eqnarray*}
\lim_{\epsilon\to 0}
\lim_{n\to\infty}
\frac{1}{n}C^{\epsilon}_{\min}(\rho^{\ot n})
\geq C_r(\rho).
\end{eqnarray*}

\end{mproof}

\section{$C^{\epsilon}_{\min}$ as an upper bound of one-shot distillable coherence }\label{apen:min_dis}

\begin{lem}\label{lem:smo_min}
Given a quantum state $\rho\in\cD(\cH)$, then for any $\mathcal{E}\in MIO$,
\begin{eqnarray}
C^{\epsilon}_{\min}(\mathcal{E}(\rho))
\leq C^{\epsilon}_{\min}(\rho).
\end{eqnarray}
\end{lem}

\begin{proof}
Let $0\leq A\leq \mathbb{I}$ and $\Tr{A\mathcal{E}(\rho)}\geq 1-\epsilon$ such that
\begin{eqnarray*}
C^{\epsilon}_{\min}(\mathcal{E}(\rho))
=\min_{\sigma\in\cI}-\log\Tr{A\sigma}
=-\log\Tr{A\sigma_*}.
\end{eqnarray*}
Then
\begin{eqnarray*}
C^{\epsilon}_{\min}(\rho)
&\geq& -\log\Tr{\mathcal{E}^\dag(A)\sigma_*}\\
&=&-\log\Tr{A\mathcal{E}(\sigma_*)}\\
&\geq& \min_{\sigma\in\cI}-\log\Tr{A\sigma}\\
&=&C^{\epsilon}_{\min}(\mathcal{E}(\rho)),
\end{eqnarray*}
where the first inequality comes from the fact that
$\Tr{\mathcal{E}^\dag(A)\rho}=\Tr{A\mathcal{E}(\rho)}\geq 1-\epsilon$ and
$0\leq\mathcal{E}^\dag(A)\leq \mathbb{I}$ as $0\leq A\leq \mathbb{I}$ and
$\mathcal{E}^\dag$ is unital.
\end{proof}

\begin{mproof}[Proof of Equation (17)]
Suppose that $\mathcal{E}$ is the optimal MIO such that
$F(\mathcal{E}(\rho), \proj{\Psi^M_+})^2\geq1-\epsilon$ with $\log M=C^{(1),\epsilon}_{D,MIO}(\rho)$.
By Lemma \ref{lem:smo_min}, we have
\begin{eqnarray*}
C^{\epsilon}_{\min}(\rho)
&\geq& C^{\epsilon}_{\min}(\mathcal{E}(\rho))\\
&=&\max_{\substack{
0\leq A\leq \mathbb{I}\\
\Tr{A\mathcal{E}(\rho)}\geq 1-\epsilon}}
\min_{\sigma\in\cI}-\log\Tr{A\sigma}\\
&\geq& \min_{\sigma\in\cI}-\log\Tr{\proj{\Psi^M_+}\sigma}\\
&=&\log M=C^{(1),\epsilon}_{D,MIO}(\rho),
\end{eqnarray*}
where the second inequality comes from the fact that
$0\leq \proj{\Psi^M_+}\leq \mathbb{I}$ and
$\Tr{\proj{\Psi^M_+}\mathcal{E}(\rho)}=F(\mathcal{E}(\rho),\proj{\Psi^M_+})^2\geq1-\epsilon$.

\end{mproof}

\end{document}